\documentclass[conference]{IEEEtran}
\IEEEoverridecommandlockouts

\usepackage{cite}

\usepackage{amsmath,amssymb,amsfonts}
\usepackage{amsthm}
\usepackage{algorithmic}
\usepackage{graphicx}
\usepackage{textcomp}
\usepackage[dvipsnames]{xcolor}
\usepackage{tcolorbox}
\tcbuselibrary{skins,breakable}
\usepackage{tikz}
\usetikzlibrary{positioning,arrows.meta,calc,fit,backgrounds,shapes.geometric,shapes.misc,decorations.pathreplacing}
\usepackage{caption}
\usepackage{booktabs}
\usepackage{array}
\usepackage{colortbl}
\usepackage{url}
\usepackage{hyperref}
\usepackage{xurl}

\newtheorem{theorem}{Theorem}
\newtheorem{lemma}{Lemma}
\newtheorem{corollary}{Corollary}
\newtheorem{proposition}{Proposition}
\newtheorem{definition}{Definition}
\newtheorem{functionality}{Functionality}

\definecolor{cDefn}{HTML}{882255}
\definecolor{cThm}{HTML}{CC6677}
\definecolor{cLem}{HTML}{44AA99}
\definecolor{cCor}{HTML}{AA4499}
\definecolor{cProp}{HTML}{DDCC77}
\definecolor{cFunc}{HTML}{117733}
\definecolor{cProto}{HTML}{BBBBBB}
\definecolor{cLink}{HTML}{1A5FB4}

\hypersetup{colorlinks=true, urlcolor=cLink, linkcolor=black,
  citecolor=black, breaklinks=true}

\tcolorboxenvironment{definition}{enhanced,
  colback=cDefn!8, colframe=cDefn!75, boxrule=0.3pt, arc=2pt,
  boxsep=2pt, left=4pt, right=4pt, top=2pt, bottom=2pt}
\tcolorboxenvironment{functionality}{enhanced,
  colback=cFunc!8, colframe=cFunc!75, boxrule=0.3pt, arc=2pt,
  boxsep=2pt, left=4pt, right=4pt, top=2pt, bottom=2pt}
\tcolorboxenvironment{theorem}{enhanced,
  colback=cThm!8, colframe=cThm!75, boxrule=0.3pt, arc=2pt,
  boxsep=2pt, left=4pt, right=4pt, top=2pt, bottom=2pt}
\tcolorboxenvironment{lemma}{enhanced,
  colback=cLem!8, colframe=cLem!75, boxrule=0.3pt, arc=2pt,
  boxsep=2pt, left=4pt, right=4pt, top=2pt, bottom=2pt}
\tcolorboxenvironment{corollary}{enhanced,
  colback=cCor!8, colframe=cCor!75, boxrule=0.3pt, arc=2pt,
  boxsep=2pt, left=4pt, right=4pt, top=2pt, bottom=2pt}
\tcolorboxenvironment{proposition}{enhanced,
  colback=cProp!8, colframe=cProp!75, boxrule=0.3pt, arc=2pt,
  boxsep=2pt, left=4pt, right=4pt, top=2pt, bottom=2pt}

\def\BibTeX{{\rm B\kern-.05em{\sc i\kern-.025em b}\kern-.08em
    T\kern-.1667em\lower.7ex\hbox{E}\kern-.125emX}}

\providecommand{\getsr}{\ensuremath{\overset{\$}{\leftarrow}}}

\providecommand{\Real}{\ensuremath{\mathsf{Real}}}
\providecommand{\Ideal}{\ensuremath{\mathsf{Ideal}}}
\providecommand{\View}{\ensuremath{\mathsf{view}}}
\providecommand{\indist}{\ensuremath{\stackrel{c}{\approx}}}

\begin{document}

\title{Reliable Homomorphic Matching for Fuzzy Labeled PSI at Scale}

\author{\IEEEauthorblockN{Erkam Uzun$^{*}$}
\thanks{$^{*}$This work extends research initiated during the author's doctoral studies at the Georgia Institute of Technology. All contributions are original and were developed independently using the author's personal time and resources, without affiliation to any current or former employer.}}

\maketitle
\thispagestyle{plain}
\pagestyle{plain}

\begin{abstract}

Fuzzy Labeled Private Set Intersection (FLPSI) lets a receiver learn the labels of enrolled records that are similar to its query, and nothing else. FLPSI can be built in several ways. Constructions based on a set-threshold reduction reach practical performance: a query matches a record when the two agree on a threshold number of components. These constructions delegate the private matching to an inner set-threshold kernel. We study its homomorphic form, which combines leveled-BFV homomorphic encryption (HE), a garbled circuit, and secret sharing to decide the match under encryption and release the record's label. We identify a composition gap in this kernel, an instance of a protocol-level problem: efficiency is bought with a per-trial false-accept probability, but one query runs a trial for every record, so the error compounds with the database size into the kernel's realization soundness error (RSE), the rate at which it accepts a query the plaintext matcher would reject. The RSE is a reliability property of the cryptographic matching layer, not the matcher's accuracy. On a spurious accept the kernel also returns a value the plaintext matcher would never produce. A sound kernel must contribute zero or negligible RSE of its own. We formalize this requirement as a composable security property, give a closed-form bound on the receiver's advantage, and close the gap with CSTPSI, a kernel that runs independent token rounds and raises the per-trial bound to a matching power. We prove CSTPSI secure in the semi-honest model. The closed-form bound sets the round count: two token rounds suffice for million-scale databases and three for billion-scale at the $10^{-6}$ engineering threshold. Our evaluation confirms the prediction. At a million records the baseline kernel's RSE reaches $100\%$ while CSTPSI holds it at $0$ in every measured configuration. CSTPSI decouples threshold-checking and upload costs from label size. For large labels at small to moderate scale it is more than $20\times$ faster than the baseline kernel, with up to 93\% less communication. It converges to the baseline only at the million-scale database size. Our implementation, with a one-command reproducibility harness, is publicly available.

\end{abstract}


\section{Introduction}\label{sec:intro}

Private set intersection (PSI) lets two parties, a sender and a receiver, compute the intersection of their private sets without revealing anything else. A long line of work~\cite{fnp2004,psz2014,kkrt2016,prty2020,rs2021,rr2022} has made exact-equality PSI fast enough for wide deployment, for example in contact discovery and compromised-credential checking. Many applications, however, need similarity instead of exact equality: biometric templates never repeat exactly across captures, and a password is risky even when it is only close to a leaked one. \emph{Fuzzy Labeled PSI} (FLPSI) serves this regime. The sender holds a database of enrolled records, each carrying a label; the receiver holds a query and learns the labels of the records that are similar to it, and nothing else, while the sender learns nothing at all.

This paper is about a gap that opens when the matching inside FLPSI is made fast enough for practice. Secure realizations of a similarity matcher rarely evaluate the closeness predicate exactly; trading exactness for efficiency at scale is the norm across secure similarity search, in PSI and beyond~\cite{uzun2021fuzzy,cfr2023,chen2020sanns}. In FLPSI the trade takes a specific form: the predicate is replaced by a cryptographic test that is exact on a true match but wrongly accepts an unrelated pair with a small per-trial probability. The bound looks harmless until one notices that a single online query is not one test but many, run across the whole database. The per-trial probability then compounds into a per-query rate that grows with the database, and at deployment scale the kernel accepts a query that matches no record. This is a soundness defect of the realization: it accepts inputs the plaintext matcher would reject, and we call the rate at which it does so the \emph{realization soundness error} (RSE). The RSE is a reliability property of the realization, not an accuracy metric, so it cannot be weighed against the matcher's intrinsic false-match rate (FMR), a plaintext floor that any realization inherits and that is out of scope here. The problem is protocol-level: it appears in any realization that makes this trade. We give the model and the attack in \S\ref{sec:threat}.

We study this RSE in the setting where it can be isolated and measured: the set-threshold reduction~\cite{uzun2021fuzzy,cfr2023}, one of the several ways to build FLPSI. It encodes each record and each query into a tuple of $N$ component items so that similar inputs agree on many components, declares a match when at least $k$ of the $N$ positions agree, and runs the private matching in a discrete inner kernel. That last property is what matters here: with the matching confined to a kernel, the kernel's RSE separates from the matcher's FMR. Within this reduction, our focus is the kernel's homomorphic form, which we call a \emph{leveled-HE-based set-threshold kernel}: a composition of leveled-BFV homomorphic encryption (HE), a garbled-circuit (GC) oblivious pseudorandom function (OPRF), and Shamir secret-share reconstruction.\footnote{The formal sections (\S\ref{sec:cstpsi-def} onward) treat the kernel as a two-party protocol and keep the standard phrasing ``protocol $\Pi_T$''.} We choose this form because one of the most practical FLPSI constructions is built on it and shows the gap most sharply. Our reference example is \mbox{STLPSI}, the inner sub-protocol of the FLPSI'21 protocol~\cite{uzun2021fuzzy}; we refer to it as the baseline kernel, or the baseline protocol, in the rest of the paper. One caveat guides the whole study. Secure-computation remedies are tied to the schema of the construction they harden, so a single fix cannot cover every realization and should not try. We therefore formalize the gap once, in the set-threshold setting, and close it on this representative kernel with a new construction, \mbox{CSTPSI}. The result is meant as a reusable template rather than a one-off patch.

The trial count is concrete in this kernel. For each query it checks every $k$-subset of the $N$ component positions in each of the $n_{\mathit{part}}$ database partitions, $\binom{N}{k}\,n_{\mathit{part}}$ independent trials in all, with $n_{\mathit{part}}$ growing linearly in the database size $D$. Each trial is exact on a true match but wrongly accepts an unrelated pair with probability $1/F$, a coincidental collision in the Shamir field of size $F$. The per-query RSE aggregates these trials, so it rises with $D$ toward one.

What the kernel returns on a spurious accept confirms the defect is one of soundness. The released value is \emph{off-curve}. It is a field element outside the label range $[0,D)$, independent of every record, so it is not a stored label and the accept leaks nothing (\S\ref{sec:far-problem} explains why). The same kernel also runs without labels, where it returns only a match bit. There a spurious accept is simply a wrong answer, even under a perfect matcher, with no off-curve value to mark it. Either way a sound kernel must have zero or negligible RSE of its own. Prior constructions either inherit this gap silently or re-tune parameters around it per deployment; none gives a composable treatment, which we survey in \S\ref{sec:background}.

We give that composable treatment in the two-party semi-honest model, with five contributions:
\begin{itemize}
  \item \textbf{A composition soundness gap, identified and modeled.} We model the set-threshold kernel's per-trial spurious accept as a deployment-scale soundness defect, and cast it as a concrete attack by a semi-honest receiver (\S\ref{sec:threat}).
  \item \textbf{A composable RSE definition and bound.} We define security for the set-threshold setting so that the kernel's added probability of accepting an unmatched query is the quantity of interest. The definition yields a closed-form bound (Lemma~\ref{lem:far-bound}) that makes the workload composition explicit, and it recovers the per-element analysis of exact-match labeled PSI as a special case.
  \item \textbf{CSTPSI: a composable validator construction, proved secure.} \mbox{CSTPSI} runs $T$ independent token rounds, each an independent validator of a candidate match, so each trial's success probability falls from $1/F$ to $1/F^{T}$ and the per-query bound becomes $\binom{N}{k}\,n_{\mathit{part}}/F^{T}$ (Theorem~\ref{thm:multi-token}). We prove it secure in the semi-honest model through a simulation-based reduction (\S\ref{sec:security}).
  \item \textbf{Protocol engineering achieves soundness restoration for free.} Amplification adds token rounds, and a naive realization would pay for each with a fresh garbled circuit and encrypted-query upload, so cost would grow with the round count and the label size. \mbox{CSTPSI} instead introduces several improvements without changing the ideal functionality or the simulation proof. The receiver (i) evaluates the garbled circuit once per session, (ii) sends the query once, and (iii) caches the modular inverses that reconstruction reuses, while the sender caches the expanded encrypted query-power window. These optimizations reduce the per-round work to plaintext-ciphertext inner products and decouple the threshold-check and upload costs from the label size (\S\ref{sec:building-blocks}).
  \item \textbf{Implementation and artifact.} A complete implementation with a one-command harness that regenerates every table and figure in this paper, publicly available.\footnote{\url{https://github.com/euzun/cstpsi}}
\end{itemize}

Our measurements confirm both the gap and the fix. We first use synthetic data that realizes a perfect matcher with zero FMR, so every accept is purely the kernel's RSE and its scaling with the database is clear: at the FLPSI'21 parameters the baseline kernel's RSE climbs from under $1\%$ at $1$K and $4.5\%$ at $10$K records to $56\%$ at $100$K and $100\%$ at $1$M, while \mbox{CSTPSI} at $T = 2$ admits no spurious accept in any measured configuration through $1$M records, below the $10^{-6}$ ($\approx\!2^{-20}$) engineering threshold; the bound shows $T = 3$ suffices at billion scale. We then run the attack on two real datasets of different scales, LFW and Deep1B, where the baseline's RSE climbs to near one at a million records ($0.999$ on Deep1B) while \mbox{CSTPSI} holds it at zero (\S\ref{sec:eval-real}). The added security costs little. With its cached garbled circuit and encrypted query, \mbox{CSTPSI} matches or outperforms the baseline across the tested range. It is more than $20\times$ faster for large labels on small-to-moderate databases, converges to within about $2\%$ of the baseline at a million records, and sends up to $93\%$ less data in the regimes where it leads.

\section{Threat Model and Assumptions}\label{sec:threat}

We state here the system this paper studies and the powers we grant the
adversary. We then draw the line the rest of the paper rests on: between
the matcher's own error floor, which the protocol cannot and should not
change, and the error the realization makes, which the protocol can and
must close. Finally, we say what counts as an attack.

\subsection{System and Adversary Model}\label{sec:threat-model}

The system is an FLPSI deployment built on the set-threshold reduction.
A \emph{sender} holds a database of enrolled records, each carrying a
label. A \emph{receiver} holds a query and wishes to learn the labels of
the records that are similar to it, and nothing else. The matching is
delegated to an inner \emph{leveled-HE-based set-threshold kernel}, whose job is
to realize the plaintext $k$-of-$N$ matcher under
encryption: a query matches a record when the two agree on at least $k$
of their $N$ components, and on a match the kernel releases the record's
label. The receiver submits one query per online interaction. In scope
is what the kernel does on each submitted query: which queries it
accepts and which values it returns. Out of scope is the choice of
representation that turns raw inputs into the $N$ components.

We adopt the standard semi-honest two-party model.
\begin{itemize}
\item Both parties are probabilistic polynomial-time
\emph{semi-honest} adversaries with security parameter
$\lambda$: each follows the protocol specification but may
attempt to learn additional information from its view of the
transcript.

\item We analyse a single execution in isolation; guarantees
under multiple concurrent executions are out of scope.

\item We rely on the standard assumptions of the building blocks
introduced in \S\ref{sec:building-blocks}: the IND-CPA security
of \mbox{BFV} leveled homomorphic encryption~\cite{seal}, the
pseudorandomness of \mbox{AES-128}~\cite{wang2017emp,aby}, and
the information-theoretic threshold property of Shamir secret
sharing~\cite{shamir1979share}.

\item The sender draws a fresh \mbox{OPRF} key per execution;
the receiver generates the HE keys and
shares only the public key with the sender.  No trusted third
party is involved.
\end{itemize}

The model makes no assumptions about the origin or domain of database
items. Rows may derive from biometric templates, attribute arrays,
medical records, or any structured input amenable to a $k$-of-$N$
component decomposition.

\subsection{Two Error Sources}\label{sec:threat-errors}

This paper separates two error sources with different owners, both
introduced in \S\ref{sec:intro}. The \emph{false-match rate} (FMR)
belongs to the matcher. It is a plaintext floor that we inherit and
leave out of scope. The \emph{realization soundness error} (RSE)
belongs to the kernel. It is the rate at which the kernel accepts a
query the matcher would reject, and it is the quantity we target. The
two are not interchangeable.

Why the RSE matters is easy to undervalue, so we state it plainly. The
kernel's one job is to realize the matcher faithfully, so that its
accept-or-reject decision can be trusted as the matcher's own. A
spurious accept breaks that trust. It cannot be weighed against the
FMR, because a large FMR is the matcher honestly reporting
chance-similar records, while a spurious accept is the kernel reporting
a match that was never there. Filtering the
values afterward does not help, because a fixed fraction of them lands
in the label range by chance and passes a range check silently. The
stakes are starkest without labels. \mbox{CSTPSI} also runs as
unlabeled fuzzy \mbox{PSI}, where the kernel returns only a match bit,
so a spurious accept is a wrong answer outright, even under a perfect
matcher, with nothing to inspect or filter. A
deployable kernel must therefore contribute zero or negligible RSE of
its own. The attack below shows why the requirement is strict, because
at scale the untreated error swamps the match decision.

\subsection{The Composition Attack on FLPSI}\label{sec:threat-attack}

We now turn this soundness gap into a
concrete attack: a semi-honest receiver, issuing only ordinary queries,
drives the kernel's realization soundness error from negligible to
near-certain as the database grows. We give first the adversary and procedure, and then what a spurious accept yields.

\paragraph{Assets and adversary.}
The asset is the soundness of the match decision: every accept should
mark a record the plaintext matcher would also accept, and every
released value should be a genuine label of such a record. The adversary
is a semi-honest receiver that submits queries through the normal
interface. It need not deviate from the protocol, and it need not hold
any input similar to an enrolled record. It only has to issue ordinary
queries and read what the kernel returns.

\paragraph{The procedure.}
The attack is nothing but the trial structure of \S\ref{sec:intro}
carried to scale. The receiver issues ordinary queries; each drives
$\binom{N}{k}\,n_{\mathit{part}}$ independent trials, and aggregated
over a database of size $D$ the per-query realization soundness error
climbs toward one. This is a multiple-comparisons effect: a per-trial
rate harmless in isolation turns near-certain across the simultaneous
tests of one query. After passing a dtabase size, the kernel accepts
essentially every query, even with no similarity to any
record.

\paragraph{Why this is a protocol-level problem.}
The compounding follows from the workload, not from any flaw unique to
one scheme. Any realization that buys efficiency with a per-trial
spurious accept, and that runs a trial for every record, carries the
same gap: the per-query rate grows with the database and the match
decision eventually fails. The set-threshold kernel is where the gap
can be isolated and measured, since the RSE there separates
cleanly from the matcher's FMR, but the shape of the problem is
general. We formalize the per-query rate as a closed-form bound
(Lemma~\ref{lem:far-bound}) and measure the break on real data, where
the baseline kernel's RSE reaches near one at a million records while
our construction holds it at zero (\S\ref{sec:eval-real}). The rest of
the paper closes this gap inside the kernel and sizes the fix to the
database scale.

\section{Related Work}\label{sec:background}

\begin{table*}[t]
  \centering
  \caption{Audit-thread positioning. ``Thr.'': $k$-of-$N$ matching;
    ``Lbl.'': label retrieval on match; ``RSE'': analyzes spurious-accept
    scaling for its threshold or label primitive.}
  \label{tab:audit-thread}
  \footnotesize
  \setlength{\tabcolsep}{3pt}
  \begin{tabular}{lllllll}
  \toprule
    Protocol & Venue / Year & Mechanism & Thr. & Lbl. & RSE & Role in audit thread \\
    \midrule
    Resende and Aranha~\cite{rdfa2018}              & FC'18        & cuckoo-filter tags     & no  & no  & per-element only     & flagged-by-successors \\
    Kales et al.~\cite{krssw2019}                   & USENIX'19    & cuckoo-filter tags     & no  & no  & per-element only     & deployed-with-bound \\
    Cong et al.~\cite{cong2021labeled}              & CCS'21       & HE polynomial eval     & no  & yes & per-execution        & recomposes deployed bound \\
    Garimella et al.~\cite{Garimella:CRYPTO:2021}   & CRYPTO'21    & OKVS amplification     & no  & no  & data-structure level & algebraic precedent ($p \mapsto p^{c}$) \\
    Uzun et al.~\cite{uzun2021fuzzy}                & USENIX'21    & HE + GC, single token  & yes & yes & no                   & no scale RSE analysis (baseline we extend) \\
    Chakraborti et al.~\cite{cfr2023}               & USENIX'23    & subsample reconciliation & yes & no & no               & no scale RSE analysis \\
    van Baarsen and Pu~\cite{baarsenpu2024fuzzyhyperballs} & EUROCRYPT'24 & DDH spatial hashing & no & no & no               & flags leakage in prior distance-aware FPSI \\
    \midrule
    CSTPSI (this work)                              & -            & HE + GC, multi-token   & yes & yes & yes (Lem.~\ref{lem:far-bound}, Thm~\ref{thm:multi-token}) & general formalization + constructive fix \\
    \bottomrule
  \end{tabular}
  \vspace{-1.5em}
\end{table*}

This section surveys prior work on private set intersection
(\mbox{PSI}) along the two axes most relevant to our construction,
\emph{fuzziness} of the matching predicate and \emph{labeling} of
the sender's records, and positions \mbox{CSTPSI} relative to
\mbox{STLPSI}, the closest prior threshold-labeled construction and
our reference baseline.  Exact-equality
\mbox{PSI}~\cite{fnp2004,psz2014,kkrt2016,prty2020,rs2021,rr2022} is
deployed at consumer scale but does not apply directly to similarity
matching, so we leave it aside and focus on the fuzzy line.

The fuzzy \mbox{PSI} (\mbox{FPSI}) literature splits along two
reductions of the closeness predicate.  The first is the
set-threshold ($k$-of-$N$) reduction of our
setting~\cite{chmielewski2008fuzzy, ye2009efficient, calapodescu2017compact, uzun2021fuzzy,cfr2023}.  The second computes the
predicate directly as a metric inequality $d_p(x, y) \le \tau$ for
an $\ell_p$ distance ($p \in \{1, 2, \infty\}$) or a Hamming
threshold~\cite{sadeghi2009efficient, erkin2009privacy,osadchy2010scifi, barni2010privacy,blanton2011secure, GRS22, gqllw2024fmap, baarsenpu2024fuzzyhyperballs,pstkz2025darot,zccbllww2025, yhwdwz2026fpsi, mwzl2026}.
Garimella et al.~(CRYPTO'22)~\cite{GRS22} suppress the realization
soundness error (RSE)
at the batched-function-secret-sharing level by repeating $\ell$
independent \mbox{bFSS} sharings ($\ell{=}280$ as instantiated).  Gao et al.~\cite{gqllw2024fmap}
achieve strictly linear communication via a ``fuzzy mapping''
assumption, with a reported $30$--$305\times$ speedup over the prior
\mbox{EUROCRYPT} 2024 baseline~\cite{baarsenpu2024fuzzyhyperballs}.  Yang et
al.~\cite{yhwdwz2026fpsi} push $\ell_p$-distance \mbox{FPSI} into
the symmetric-key regime (a reported $12$--$145\times$ improvement),
and report a false-positive leakage in a recent linear-complexity
\mbox{AHE}-based construction~\cite{dang2025fpsi}.  Meng et al.~\cite{mwzl2026} target unbalanced
biometric search under $L_\infty$ with sub-linear communication,
and further works push structure-aware and distance-aware-OT
directions~\cite{pstkz2025darot,zccbllww2025}.  None of these
addresses the fuzzy labeled \mbox{PSI} (\mbox{FLPSI}) setting we work
in, where the receiver learns the labels of similar records, though
all attack the same fuzzy-intersection problem under different
efficiency and metric profiles.

Several primitives sit near \mbox{CSTPSI}'s functionality but solve
a different problem.
\emph{Cardinality-threshold PSI}~\cite{ghosh2019threshold,badrinarayanan2021multiparty}
returns the exact-equality intersection only when its cardinality
exceeds a threshold $t$ and otherwise returns nothing; the
threshold there is over $|X \cap Y|$ across the entire sets, not
over component agreement within a single fuzzy item.
\emph{Labeled \mbox{PSI} from HE}
(\mbox{APSI}~\cite{chen2018labeled,cong2021labeled}) returns labels
for exact-match items; the per-element $1/F$ false-positive analysis
at the heart of its cuckoo/bin design is the lineage from which the
bound of Lemma~\ref{lem:far-bound} generalizes.  \emph{Searchable
encryption}, \emph{private information retrieval}, and
\emph{circuit-PSI}~\cite{chandran2022circuit,Pinkas:EUROCRYPT:2019}
share parts of the toolkit but address orthogonal functionalities:
encrypted-index query, single-row recovery, and arbitrary-circuit
set evaluation, respectively.

The thread that motivates this work cuts across all of these lines,
and Table~\ref{tab:audit-thread} collects it. The works in the
table target different problem settings, so the comparison is
qualitative positioning rather than a performance ranking. The common pattern is a
per-trial error that a single query runs against every record, so its
small per-trial probability compounds with the database size
(\S\ref{sec:intro}).  Several constructions do not analyze this
scaling: the unbalanced \mbox{PSI} of Resende and Aranha, built on
cuckoo-filter tags~\cite{rdfa2018}, the contact-discovery protocols
of Kales et al.~\cite{krssw2019}, the distance-aware \mbox{PSI} of
Chakraborti et al.~\cite{cfr2023}, and the $k$-out-of-$N$ labeled
fuzzy \mbox{PSI} of Uzun et al.~\cite{uzun2021fuzzy}, the baseline
kernel we extend.  Others control the error per deployment or per
execution: the \mbox{APSI} line sizes its hash range and cuckoo
bounds against the workload~\cite{chen2018labeled,cong2021labeled},
and modern \mbox{OKVS-PSI} absorbs it into the field
size~\cite{rr2022}.  A third group, including van Baarsen and
Pu~\cite{baarsenpu2024fuzzyhyperballs} and Yang et
al.~\cite{yhwdwz2026fpsi}, frames the error as information disclosure
and flags false-positive leakage in prior fuzzy-\mbox{PSI}
constructions.  
Closest to a fix, Garimella et al.~(CRYPTO'21)~\cite{Garimella:CRYPTO:2021}
amplify an \mbox{OKVS} overfitting probability from $p$ to $p^{c}$
inside the encoding structure, and the $\ell$-fold \mbox{bFSS}
repetition of Garimella et al.~(CRYPTO'22)~\cite{GRS22} above is the
same move at the function-secret-sharing layer.  Both strengthen one trial
inside a sub-primitive.  This is the wrong layer for the soundness
gap: neither composes the guarantee across the many trials a query
runs against a full database, and neither exposes the amplification
as a tunable protocol parameter with a matching security definition.
\mbox{CSTPSI} supplies exactly these two missing pieces, which makes
it, to our reading, the first composable treatment of the gap,
demonstrated in the threshold-labeled setting.

\section{Preliminaries}\label{sec:overview}

This section introduces the protocol class we analyse, formalizes
the realization soundness error (RSE) composition problem that any
protocol in this class faces at deployment scale, and rules out the
obvious parameter-level remedies.  Our solution, \mbox{CSTPSI}, is introduced in \S\ref{sec:cstpsi}
and formalized in \S\ref{sec:cstpsi-def}.

\subsection{Set-Threshold Fuzzy PSI ($k$-of-$N$ matching)}\label{sec:stlpsi-primitive}

We work over a protocol class commonly called \emph{Set-Threshold
fuzzy PSI} or \emph{$k$-of-$N$ matching}: each row of the sender's
database and each receiver query is an $N$-tuple of items from a
fixed domain $\mathcal{D}$, and a match is declared when at least
$k$ of $N$ item positions agree under the closeness predicate
$\mathsf{Cl}_\tau$~\cite{uzun2021fuzzy,cfr2023}.
Table~\ref{tab:notation} collects the symbols used throughout the
rest of the paper.  As a representative member of this class we
take \mbox{STLPSI}, the labeled fuzzy \mbox{PSI}
primitive of Uzun et al.~\cite{uzun2021fuzzy}, as the baseline protocol
since it carries the RSE composition problem we analyse in this
section in its sharpest form.

\begin{table}[t]
\centering
\caption{Notation and parameters used throughout the paper. Values are
  fixed across experiments; ranges are swept in \S\ref{sec:eval}.}
\label{tab:notation}
\small
\setlength{\tabcolsep}{4pt}
\begin{tabular}{@{}l p{0.80\columnwidth}@{}}
\toprule
Symbol & Meaning (fixed value, where applicable) \\
\midrule
$N$ & Component items per database row ($64$) \\
$k$ & Matching threshold of agreeing items ($2$) \\
$T$ & Independent token rounds ($1, 2, 3$) \\
$K$ & Label-chunk rounds for $23$-bit and $16$/$32$/$64$-byte labels ($1, 6, 12, 23$) \\
$D$ & Database size in rows ($1$K--$1$M) \\
$s_{\mathit{part}}$ & Maximum rows per partition ($32$) \\
$n_{\mathit{part}}$ & Number of partitions, $\lceil D / s_{\mathit{part}} \rceil$ \\
$F$ & BFV plaintext modulus and Shamir field $\mathbb{F}_F$ ($8519681$, ${\sim}23$-bit) \\
$m$ & BFV polynomial-modulus degree ($4096$) \\
$\mathcal{D}, \mathcal{L}$ & Item domain and label space \\
$X$ & Sender's items $\{x_1, \dots, x_D\}$, $x_e \in \mathcal{D}$ \\
$\ell_{x_e}$ & Label of sender item $x_e$ (in $\mathcal{L}$) \\
$Y$ & Receiver's query $(y_1, \dots, y_N) \in \mathcal{D}^N$ \\
$\mathsf{Cl}_\tau$ & Closeness predicate \\
$\kappa$ & Sender's AES-OPRF key \\
$b_j$ & Blinded query item: $\mathrm{AES}_\kappa(y_j) \bmod F$ \\
$P_{t,e}(x)$ & Per-row Shamir polynomial in round $t$ \\
$f_{t,i,p}(x)$ & Per-partition interpolating polynomial \\
$\mathsf{KR}$ & Shamir secret reconstruction (see \S\ref{sec:bb-shamir}) \\
\bottomrule
\end{tabular}
\end{table}

Attaching a label to each row turns this matching primitive into
\emph{Fuzzy Labeled PSI} (\mbox{FLPSI}, \S\ref{sec:intro}), the
functionality this paper studies: the receiver learns the labels
of the rows it matches, and nothing else.  The sender inputs a
database
$\mathit{Db} = \{(x_e, \ell_{x_e})\}_{e=1}^{D}$ in which each row
$x_e$ is an $N$-tuple of items and each row carries an opaque
label $\ell_{x_e} \in \mathbb{F}_F$.  The receiver inputs a
single fuzzy query $Y = (y_1, \dots, y_N)$ from $\mathcal{D}^N$.
The functionality returns to the receiver the label $\ell_{x_e}$
of every row $e$ that agrees with the query on at least $k$ of $N$
item positions, and $\bot$ to the sender.  Labels are
\emph{opaque} from \mbox{STLPSI}'s perspective: the primitive
treats them as arbitrary field elements, so any structural
decoration the labels need (a validation prefix, multiple
polynomials for cross-validation, additional verification rounds)
must be assembled by the caller before invoking the primitive and
stripped after the output.

Each row's $N$ items are blinded
under an OPRF ($b_{e,i}$ on the sender side, $b_j$ on the receiver
side); per item position $i$, the sender packs at most
$s_{\mathit{part}}$ rows into a single interpolating polynomial
$f_{t,i,p}$ of degree $s_{\mathit{part}}-1$ over $\mathbb{F}_F$,
evaluated at the receiver's blinded query inputs in a single
homomorphic round; the receiver recovers each polynomial value
and tries every $\binom{N}{k}$ unordered $k$-subset of item
positions per partition via $k$-point Lagrange reconstruction at
$x = 0$\footnote{Practical \mbox{STLPSI} deployments instantiate
the threshold at $k = 2$, yielding a $2$-point Lagrange
operation; the bounds in this section apply at general $k$.}.
Each token round uses a fresh degree-$(k-1)$ Shamir polynomial
$P_{t,e}$ per row whose constant term is the public token $0$;
reconstructions that hit $0$ trigger a candidate match.  The full
construction with all design rationale is in
\S\ref{sec:building-blocks}.

\subsection{The RSE Composition Problem}\label{sec:far-problem}

To our knowledge, \mbox{FLPSI} protocols
evaluated to date have not analyzed RSE scaling at deployment
scale (\S\ref{sec:background}).  We observe that the per-query RSE of this class scales
linearly with the number of database partitions
$n_{\mathit{part}} = \lceil D / s_{\mathit{part}}\rceil$, and
therefore degrades rapidly as $D$ grows into the database scales
that motivate fuzzy \mbox{PSI} in the first place: biometric
watchlist matching at the hundred-thousand to ten-million scale,
and credential-leak monitoring well above that scale.

The per-query RSE is a soundness defect, not just an accuracy
nuisance.  On a non-matching pair the kernel reconstructs an
\emph{off-curve} value, uniform in $\mathbb{F}_F$ and independent of
any row's payload, so it is not a stored label but a value the
plaintext matcher would never return.  This is distinct from
stored-label exposure, which happens only on a genuine $\ge k$ match,
the matcher's intrinsic false-match rate that lies outside our scope.
Lemma~\ref{lem:far-bound} formalizes the scaling and bounds the
probability of such a wrong accept; its proof shows the reconstructed
value is uniform.

\begin{figure}
  \center
  \resizebox{0.98\columnwidth}{!}{
\begin{lemma}[Composition Bound on the RSE]\label{lem:far-bound}
Let $\Pi_T$ be a labeled fuzzy \mbox{PSI} protocol with
$T \ge 1$ independent Shamir-based token rounds and parameters
$(N, k, s_{\mathit{part}}, F)$.  For a database of $D$ entries
with $n_{\mathit{part}} = \lceil D/s_{\mathit{part}}\rceil$, and
a random query $Y$ sharing fewer than $k$ items with every
enrolled entry, the probability that $\Pi_T$ wrongly accepts and
returns any value is at most
\begin{equation}
\Pr[\,\mathrm{accept}\,] \;\le\; \binom{N}{k}\,n_{\mathit{part}}\,\big/\,F^{T}
+ T \cdot \mathsf{Adv}^{\mathrm{prf}}_{\mathrm{AES}}(\mathcal{B}),
\label{eq:far-T}
\end{equation}
where $\mathsf{Adv}^{\mathrm{prf}}_{\mathrm{AES}}(\mathcal{B})$
is the \mbox{AES} pseudorandomness advantage of an efficient
adversary $\mathcal{B}$; under standard assumptions this term is
bounded by $2^{-128}$ up to negligible constants.  At $T = 1$,
the dominant term is tight to first order when
$\binom{N}{k}\,n_{\mathit{part}}/F \ll 1$:
$\binom{N}{k}\,n_{\mathit{part}}/F = 1 - (1 - 1/F)^{\binom{N}{k}\,n_{\mathit{part}}}
+ O\!\big((\binom{N}{k}\,n_{\mathit{part}}/F)^{2}\big)$.
Since $\binom{N}{k}\,n_{\mathit{part}}$ grows with $D$, this ratio
is not small at deployment scale, where the linear form
overestimates and the exact
$1 - (1 - 1/F)^{\binom{N}{k}\,n_{\mathit{part}}}$ should be used
(\S\ref{sec:eval}).
\end{lemma}
  }
\vspace{-1.5em}
\end{figure}

\begin{proof}[Proof sketch]
Replace $\mathrm{AES}_{\kappa}$ with a random function $R$ at the
cost of $T \cdot \mathsf{Adv}^{\mathrm{prf}}_{\mathrm{AES}}(\mathcal{B})$
by a standard \mbox{PRF/PRP} hybrid argument across the $T$ token
rounds.  Conditioned on the $R$-output, fix a partition $p$
containing no entry with $\ge k$ matches against $Y$, and fix a
$k$-subset $I = \{i_1, \ldots, i_k\}$ of item positions.  In round
$t$, the sender's per-entry Shamir polynomial $P_{t,e}$ is sampled
with fresh random coefficients, so the partition polynomial
$f_{t,i,p}$ interpolates the partition through $s_{\mathit{part}}$
ordinate values $\{P_{t,e}(i)\}_{e \in p}$ that are independent
uniform draws over $\mathbb{F}_F$ (each is a fixed nonzero
linear combination of the row's fresh secret-zero Shamir
coefficients, which reduces to a nonzero scalar multiple of the
single fresh coefficient when $k = 2$).  Evaluated
at any point $b$ outside the $s_{\mathit{part}}$ interpolation
$x$-coordinates, $f_{t,i,p}(b)$ is a fixed nonzero linear
combination of those independent uniform values, hence itself
uniform in $\mathbb{F}_F$.  By hypothesis no entry in $p$ shares
$\ge k$ items with $Y$, so for every entry $e \in p$ at least one
position $i_j \in I$ is a non-node ($b_{i_j} \ne b_{e,i_j}$); the
$k$ values fed to the reconstruction therefore include at least
one such uniform value.  The
modular-bias and
zero-collision terms (the $0 \mapsto 1$ remap induces a per-value
bias of $2/F$ on the value $1$; the $2^{128} \to F$ reduction
induces a $O(F/2^{128}) \approx 10^{-32}$ bias) are absorbed into
the asymptotic.  The $k$-point Lagrange reconstruction at $x = 0$
is an affine combination of these $k$ values that gives nonzero
weight to a uniform non-node value, and is therefore uniform in
$\mathbb{F}_F$; equality
with the public token $0$ occurs with probability $1/F$.  Across
$T$ rounds with bin-local fresh coefficient samples, the
per-round events are independent, giving $(1/F)^{T}$.
This per-round independence rests on the Shamir coefficients being
sampled independently across rounds, an assumption on the randomness
source that Theorem~\ref{thm:multi-token} makes explicit.  A union bound over the $\binom{N}{k}$
candidate $k$-subsets in each of $n_{\mathit{part}}$ partitions
yields~(\ref{eq:far-T}); the union bound itself requires only
sub-additivity, not independence across bins.
\qed
\end{proof}

Lemma~\ref{lem:far-bound} generalizes the per-element
false-positive analysis of prior exact-matching PSI work, e.g.\
\mbox{APSI}~\cite{chen2018labeled}, which establishes $1/F$ per
query item per bin: setting $N = k = n_{\mathit{part}} = 1$
recovers their bound from ours.  The novelty here is lifting that
single-element analysis to the threshold setting, where the
combinatorial factor $\binom{N}{k}$ and the bin amplification
$n_{\mathit{part}}$ jointly drive RSE linearly in $D$.  Our
measurements over a reference codebase confirm this qualitative
shape; the full configuration table and the comparison against
the exact bound appear in \S\ref{sec:eval}.

\subsection{Why Existing Remedies Fall Short}\label{sec:far-why-not-F}

Two parameter-level remedies suggest themselves before a
structural change: (i) growing the partition size
$s_{\mathit{part}}$ to shrink the count
$n_{\mathit{part}} = \lceil D / s_{\mathit{part}}\rceil$, and (ii)
enlarging the plaintext modulus $F$ to shrink the per-trial $1/F$
probability.  Both target Lemma~\ref{lem:far-bound} directly, but
neither is viable at practical labeled fuzzy \mbox{PSI} parameters.

\paragraph{Growing the partition size $s_{\mathit{part}}$.}
Doubling $s_{\mathit{part}}$ halves $n_{\mathit{part}}$ and the
bound, but the per-partition interpolating polynomial has degree
$s_{\mathit{part}}-1$, so communication and homomorphic-evaluation
cost grow linearly in $s_{\mathit{part}}$ as well.  The trade-off is
linear-for-linear, and a cryptographically negligible bound forces
$s_{\mathit{part}}$ to a value where cross-partition amortization
collapses.

\paragraph{Enlarging the plaintext modulus $F$.}  In
\mbox{BFV}-style encryption the plaintext modulus must be prime and
satisfy $F \equiv 1 \pmod{2\,m}$ for \mbox{SIMD} (single-instruction
multiple-data) batch encoding, where $m$ is the polynomial-modulus degree.  At the $m$ used in
practical labeled-\mbox{PSI} deployments the prime budget at
$128$-bit security is exhausted near a fixed bit-width; enlarging
$F$ further forces $m$ to the next power of two, roughly doubling
per-operation evaluation cost and ciphertext size at the same
depth.  This more than offsets the improvement to the bound at every
practical label size (concrete costs in \S\ref{sec:eval}).

The structural fix therefore lies elsewhere: re-folding the
caller/primitive boundary so the amplification round count is a
parameter of the primitive, not a re-orchestration of its caller.
We introduce this fix as \mbox{CSTPSI} in \S\ref{sec:cstpsi}.

\section{Our Solution: CSTPSI}\label{sec:cstpsi}\label{sec:cstpsi-design}

We introduce \emph{Composable Set-Threshold PSI} (\mbox{CSTPSI}),
the primitive that takes the realization soundness error (RSE)
mitigation inside its boundary
by absorbing the caller-managed share-construction preamble into
the protocol itself.  Five design choices drive the construction:
where share construction lives, how many token rounds the
primitive admits, how the garbled-circuit OPRF interacts with
the round count, how the receiver's encrypted query material
is amortized across rounds, and how the receiver reuses its
reconstruction work across rounds.

\paragraph{Taking share construction into the protocol.}
Set-Threshold fuzzy \mbox{PSI} primitives that treat labels as
opaque field elements delegate share construction to the caller,
which makes any structural change to the share machinery (adding
verification rounds for amplification, cross-validation
encodings, longer-label transport) a caller-side surgery rather
than a parameter change.  \mbox{CSTPSI} accepts raw labels at its
interface and constructs Shamir
polynomials~\cite{shamir1979share} internally, parameterized by
a configurable round count $T$.  Variations on the share
machinery then become parameter choices.

\paragraph{The role of $T$ (multi-token amplification).}
\mbox{CSTPSI} exposes $T$, the number of independent token
rounds, as a first-class protocol parameter.  Each added token
round acts as an independent \emph{validator} of a candidate
match: the round re-checks the same candidate with fresh Shamir
coefficients, so a spurious accept must fool all $T$ validators
at once.  Fresh coefficients per round drive the per-trial
spurious-accept probability geometrically to $1/F^{T}$.  The
validator lives at the protocol-round level, so it is composable
across the whole query workload, and we \emph{size} the number of
validators to the database scale through the closed-form bound
(Lemma~\ref{lem:far-bound}).  Corollary~\ref{cor:sufficient-T} gives
the sufficient $T$ for any target scale, and $T = 2$ already makes the
per-query bound engineering-negligible at million-scale $D$
(\S\ref{sec:sec-amplification}).
The full amplification statement is
Theorem~\ref{thm:multi-token}.

\paragraph{Once-per-session GC-OPRF.}
The OPRF used to blind query and database items is realized as a
two-party garbled-circuit AES evaluation, which mutually conceals
the query items from the sender and the key from the receiver.
A naive realization runs this circuit per round, inflating the
receiver's online cost by a factor $T + K$ and making the GC
step the dominant component of online latency at any reasonable
$K$.  \mbox{CSTPSI} executes the garbled
circuit~\cite{wang2017emp} only \emph{once per online session},
since the blinded queries $b_j = \mathrm{AES}_\kappa(y_j) \bmod F$
are deterministic in the query items and the sender's key,
hence independent of the round.  The general principle of
minimizing protocol-conversion steps in a mixed
arithmetic--boolean computation traces to \mbox{ABY}~\cite{aby}.

\paragraph{Once-per-session encrypted query powers with sender caching.}
The sender's per-partition polynomial evaluation
$f_{t,i,p}(b_j) = \sum_{\ell=0}^{d-1} c_{t,i,p,\ell} \cdot
\mathrm{Enc}(b_j^\ell)$, with $d = s_{\mathit{part}}$, requires
ciphertexts $\mathrm{Enc}(b_j^\ell)$ for every $\ell \in [0, d)$.
Following the baby-step--giant-step (BSGS) windowing scheme used
in APSI-style labeled \mbox{PSI}~\cite{chen2018labeled,cong2021labeled,uzun2021fuzzy}, the
receiver transmits only a sparse window of $O(\sqrt{d})$
encrypted powers; the sender expands this window to the full
power set via $O(\sqrt{d})$ ciphertext-ciphertext multiplications
at session setup, consuming a small constant number of \mbox{BFV}
multiplicative levels.  \mbox{CSTPSI} then \emph{caches} the
expanded power bundle on the sender side for the duration of the
online session, so the per-round sender work reduces to
plaintext-ciphertext (coefficients $\times$ power) inner products with
no per-round ciphertext-ciphertext multiplications. The expansion cost is paid once and amortized over all $T + K$ rounds.
\footnote{The trade-off between \mbox{BSGS} and a
full-receiver-side pre-encryption variants depends on hardware
and deployment regime (network bandwidth vs computation);
Appendix~\ref{app:query-powers} describes both variants.}

\paragraph{Cached reconstruction on the receiver.}
At reconstruction the receiver runs a two-point Lagrange combination
at $x = 0$ for each candidate pair $\{i, j\}$ in a partition.  Its
weights depend only on the blinded query coordinates $b_i, b_j$, which
the single garbled circuit fixes for the whole session, so they are
identical across all $T + K$ rounds.  \mbox{CSTPSI} computes the
modular inverses in these weights once per pair and reuses them across
every round, replacing a per-round field inversion with a cached
lookup.

\mbox{CSTPSI} inherits the semi-honest two-party model and
building-block assumptions of \S\ref{sec:threat}.  Within that model
the $T$ token rounds also blunt a receiver who probes with crafted
queries to force a spurious accept, since each candidate must then
collide in all $T$ rounds rather than one (probability $1/F^{T}$, not
$1/F$); a sender that omits entries merely substitutes inputs and is
not a separate attack.

\section{Formal Definition of CSTPSI}\label{sec:cstpsi-def}

We formally define the \mbox{CSTPSI} primitive: its syntax, the
ideal functionality \(\mathcal{F}_{\mathrm{CSTPSI}}\) it realizes,
and its correctness clause. We start from a closeness abstraction
(Definition~\ref{def:closeness-domain}) that fixes the item-level
fuzzy-matching relation, then give the protocol syntax
(Definition~\ref{def:cstpsi-syntax}) over
it~\cite{boldyreva2014fuzzy, uzun2021fuzzy}.

\begin{figure}[h]
  \center
  \resizebox{0.98\columnwidth}{!}{
  \begin{definition}[Closeness Domain]\label{def:closeness-domain}
  A \emph{closeness domain} for set-threshold parameters
  \((N, k)\) with \(1 \le k \le N\) is a pair
  \(\Lambda = (\mathcal{D}, \mathsf{Cl}_\tau)\), where
  \(\mathcal{D}\) is an item set and
  \(\mathsf{Cl}_\tau \colon \mathcal{D}^N \times \mathcal{D}^N \to
  \{0, 1\}\) is the symmetric set-threshold closeness function
  defined by
  \[
    \mathsf{Cl}_\tau(Y, X) \;=\; 1
    \;\;\iff\;\;
    \bigl|\{\,j \in [N] : y_j = x_j\,\}\bigr| \;\ge\; k.
  \]
  \end{definition}
  }
  \vspace{-0.2em}
  \caption*{We read \(\mathsf{Cl}_\tau(Y, X) = 1\) as ``\(Y\) is close to
\(X\)'': the two \(N\)-tuples agree on at least \(k\) of the
\(N\) item positions.}
\vspace{-1.5em}
\end{figure}

\begin{figure}[h]
\center
  \resizebox{0.98\columnwidth}{!}{
\begin{definition}[Composable Set-Threshold PSI -- CSTPSI]\label{def:cstpsi-syntax}
A \emph{CSTPSI} protocol \(\Pi_T\) is a two-party interactive
protocol between a sender \(S\) and a receiver \(R\), parameterized
by a closeness domain
\(\Lambda = (\mathcal{D}, \mathsf{Cl}_\tau)\) per
Def.~\ref{def:closeness-domain}, a label space
\(\mathcal{L}\), the set-threshold parameters \((N, k)\) with
\(1 \le k \le N\), the token-round count \(T \ge 1\), and a
security parameter \(\lambda\).
\smallskip

\noindent\textit{Inputs.}  \(S\) holds a labeled database
\(\mathit{Db} = \{(x_e, \ell_{x_e})\}_{e=1}^{D}\) with each
\(x_e \in \mathcal{D}^{N}\) an \(N\)-tuple of component items and
each \(\ell_{x_e} \in \mathcal{L}\) the row's payload label.
\(R\) holds a query \(Y \in \mathcal{D}^{N}\).  Both parties agree
on \((\Lambda, \mathcal{L}, N, k, T, \lambda)\) and on any public
parameters the concrete realization requires.
\smallskip

\noindent\textit{Outputs.}  \(R\) outputs the multiset
\(\mathcal{R} = \{\,\ell_{x_e} : \mathsf{Cl}_\tau(Y, x_e) = 1\,\}\);
\(S\) outputs \(\bot\).
\end{definition}
  }
\end{figure}

\textbf{Ideal functionality.} The CSTPSI protocol \(\Pi_T\) realizes
\(\mathcal{F}_{\mathrm{CSTPSI}}\) semi-honestly with the simulator
gap quantified in Theorem~\ref{thm:main-sec}.
The token-round count \(T\) is a first-class protocol parameter, not a
tuning knob hidden inside the realization. The protocols
\(\Pi_1, \Pi_2, \dots\) all realize this same functionality, but with
\emph{different} simulator gaps; an operator picks \(T\) to meet a target
gap at a target deployment scale.

\begin{figure}[h]
\center
  \resizebox{0.98\columnwidth}{!}{
\begin{functionality}[Ideal functionality \(\mathcal{F}_{\mathrm{CSTPSI}}\)]\label{def:F-cstpsi}
\textit{Parameters.} Structural parameters \((N, k)\) with
\(1 \le k \le N\); closeness setting
\((\mathcal{D}, \mathsf{Cl}_\tau)\); label space \(\mathcal{L}\).

\noindent\textit{Inputs.} From \(S\), a labeled database
\(\mathit{Db} = \{(x_e, \ell_{x_e})\}_{e=1}^{D}\) with
\(x_e \in \mathcal{D}^{N}\) and \(\ell_{x_e} \in \mathcal{L}\); from
\(R\), a query \(Y = (y_1, \dots, y_N) \in \mathcal{D}^{N}\).

\noindent\textit{Computation.} For each entry \(e\), mark it a match iff
\(\mathsf{Cl}_\tau(Y, x_e) = 1\).

\noindent\textit{Outputs.} To \(R\), the multiset
\(\mathcal{R} = \{\,\ell_{x_e} : \mathsf{Cl}_\tau(Y, x_e) = 1\,\}\); to
\(S\), \(\bot\).

\noindent\textit{Leakage.} None beyond what \(\mathcal{R}\) and
\(|Y|, |\mathit{Db}|\) reveal: no information on non-matching
entries and no per-component side channel.
\end{functionality}
  }
\end{figure}

\begin{figure}[!h]
  \center
  \resizebox{0.98\columnwidth}{!}{
\begin{definition}[CSTPSI correctness]\label{def:cstpsi-correctness}
A CSTPSI protocol \(\Pi_T\) is \emph{correct} if, for every database
\(\mathit{Db}\) and every query \(Y\) such that there exists an
enrolled entry \(e\) with \(\mathsf{Cl}_\tau(Y, x_e) = 1\), the output multiset
\(\mathcal{R}_{\Pi_T}\) computed by the honest execution of
\(\Pi_T\) contains \(\ell_{x_e}\) with probability \(1\).
Equivalently, the kernel false-reject rate of \(\Pi_T\), conditioned
on a true \(k\)-of-\(N\) component agreement, is zero (the
representation's recall miss reported in \S\ref{sec:eval-real} is a
plaintext property of the sketch, identical for the matcher,
\mbox{STLPSI}, and \mbox{CSTPSI}).
\end{definition}}
\vspace{-2em}
\end{figure}

\begin{figure*}[t]
\centering
\resizebox{0.96\width}{!}{
\begin{tcolorbox}[
  enhanced, colback=cProto!8, colframe=cProto!75,
  boxrule=0.3pt, arc=2pt,
  left=6pt, right=6pt, top=4pt, bottom=4pt]
\textbf{Inputs.} \emph{Sender} S inputs $\mathit{Db} = \{(x_e,
\ell_{x_e})\}_{e=1}^{D}$, where each $x_e \in \mathcal{D}^{N}$ is an $N$-tuple of items and each $\ell_{x_e} \in \mathcal{L}$.  \emph{Receiver}
R inputs $Y = (y_1, \dots, y_N)$ with $y_j \in \mathcal{D}$.  Both
parties agree on the parameters $(N, k, T, F, s_{\mathit{part}})$
and on \mbox{BFV} public parameters.

\begin{enumerate}
  \item \textbf{[Init]}\label{step:init} S samples an AES-128 OPRF
    key $\kappa \getsr \{0,1\}^{128}$.

  \item \textbf{[S Offline: partition]}\label{step:partition}
    S partitions $\mathit{Db}$ into $n_{\mathit{part}} = \lceil D /
    s_{\mathit{part}} \rceil$ partitions of at most
    $s_{\mathit{part}}$ rows each, rejecting within-partition duplicate
    items so step~\ref{step:interp}'s interpolation is well-defined.

  \item \textbf{[S Offline: blind]}\label{step:blind} For each
    row $e$ and each item position $i \in [N]$, S computes
    $b_{e,i} = \mathrm{AES}_\kappa(x_{e,i}) \bmod F$ (with the
    convention $0 \mapsto 1$) and replaces $x_{e,i}$ with $b_{e,i}$
    in place.

  \item \textbf{[S Offline: share]}\label{step:share} For each
    row $e$ and each round $t \in \{0, \dots, T + K - 1\}$, S draws
    a fresh degree-$(k-1)$ Shamir polynomial $P_{t,e}(x) = s_{t,e}
    + a_{t,e}\,x$ with $a_{t,e} \getsr \mathbb{F}_F$, where
    $s_{t,e} = 0$ for $t \in [0, T)$ and $s_{t,e}$ is the
    $(t{-}T{+}1)$-th chunk of $\ell_{x_e}$ for $t \in [T, T + K)$,
    and evaluates $P_{t,e}(i)$ for each $i \in [N]$.

  \item \textbf{[S Offline: interpolate \& pack]}\label{step:interp}
    For each round $t$, query position $i$, and partition $p$, S
    interpolates the unique degree-$(d{-}1)$ polynomial $f_{t,i,p}$
    with $d = s_{\mathit{part}}$ through the points
    $\{(b_{e,i}, P_{t,e}(i)) : e \in p\}$ over $\mathbb{F}_F$ and
    SIMD-batches the coefficients of $m / N$ partitions into one
    \mbox{SEAL} plaintext.

  \item \textbf{[GC-OPRF: 1-GC]}\label{step:gc} R and S run a
    single two-party AES-128 garbled circuit~\cite{wang2017emp}: R
    inputs $(y_1, \dots, y_N)$, S inputs $\kappa$, R learns
    $b_j = \mathrm{AES}_\kappa(y_j) \bmod F$ for each $j \in [N]$
    (with the convention $0 \mapsto 1$), and S learns $\bot$.  This
    step is executed exactly once per online session and is reused
    across all $T + K$ subsequent rounds.

  \item \textbf{[Online Query]}\label{step:onlinequery}
    For each $j \in [N]$, R homomorphically encrypts a
    BSGS window of $O(\sqrt{d})$ powers of $b_j$ and
    transmits it once; S expands it to the full power set
    $\bigl\{\mathrm{Enc}(b_j^{\ell})\bigr\}_{\ell \in [0, d)}$ via
    $O(\sqrt{d})$ ciphertext multiplications and \emph{caches} the
    bundle, reused by all $T + K$ rounds without re-expansion.

  \item \textbf{[Per-Round Homomorphic Evaluation]}\label{step:hom}
    For each round $t \in \{0, \dots, T + K - 1\}$ in turn: for
    each $(i, p)$, S homomorphically evaluates the partition
    polynomial $\mathrm{Enc}\bigl(f_{t,i,p}(b_j)\bigr) =
    \sum_{\ell=0}^{d-1} c_{t,i,p,\ell} \cdot \mathrm{Enc}(b_j^\ell)$
    in batched form, and returns $k$ \mbox{SEAL} ciphertexts per
    round (one per query position $i, j$ consumed by
    reconstruction).  R decrypts and stores the round-$t$
    evaluations $f_{t,i,p}(b_j)$ for every $(i, j, p)$.

  \item \textbf{[Receiver Reconstruction: token check]}\label{step:reconstruct-token}
    For each partition $p$ and each unordered pair $\{i, j\} \in
    \binom{[N]}{k}$, R computes $\mathsf{KR}(i, j, t)$
    with the two-point Lagrange weights at $x = 0$.  R records the
    pair $(p, \{i, j\})$ as a token candidate iff
    $\mathsf{KR}(i, j, t) = 0$ for \emph{every} token round
    $t \in [0, T)$.

  \item \textbf{[Receiver Reconstruction: label
    recovery]}\label{step:reconstruct-label} For each token
    candidate $(p, \{i, j\})$, R computes $\mathsf{KR}(i, j, t)$
    for every label round $t \in [T, T + K)$ to recover the $K$
    chunks of the label, and reassembles them into
    $\ell_{x_e} \in \mathcal{L}$.
\end{enumerate}
\vspace{-0.2\baselineskip}

\textbf{Output.} R outputs the multiset $\{\ell_{x_e} : e
\text{ matched}\}$.  S outputs $\bot$.
\end{tcolorbox}
}
\caption{\mbox{CSTPSI} protocol $\Pi_T$.  Lemma~\ref{lem:far-bound}
gives the realization soundness error (RSE) bound at any $T$; the
multi-token-round mitigation drives the bound to negligible for
$T \ge 2$ at the deployment scales of \S\ref{sec:eval}.}
\label{fig:cstpsi-protocol}
\vspace{-1.2em}
\end{figure*}

\textbf{Correctness.} The correctness clause is deterministic with respect to true
matches.  We deliberately do not adopt an \(\varepsilon\)-correctness
clause that would allow the protocol to be wrong on spurious accepts
with negligible probability~\cite{uzun2021fuzzy}.
An \(\varepsilon\)-correctness clause creates a loophole: it
discharges the obligation to account for spurious accepts in either the
correctness analysis (where they would surface as deviations from
\(\mathcal{F}_{\mathrm{CSTPSI}}\)) or the security analysis (where
they would surface as simulator-gap mass).  We instead treat the
spurious-accept event as a \emph{soundness} event: a quantified
deviation from \(\mathcal{F}_{\mathrm{CSTPSI}}\)'s authoritative
output, in which the real protocol accepts a query the ideal
functionality would reject and returns a value the ideal
functionality would never produce.  Lemma~\ref{lem:far-bound} and
Theorem~\ref{thm:multi-token} bound this event, and it appears as a
soundness contribution to the simulator gap of
Theorem~\ref{thm:main-sec}.  Accounting for it this way closes the
loophole and composes cleanly with downstream applications that
treat \(\mathcal{F}_{\mathrm{CSTPSI}}\)'s output as authoritative.

The CSTPSI security definition and the proof that \(\Pi_T\)
realizes \(\mathcal{F}_{\mathrm{CSTPSI}}\) are presented in
\S\ref{sec:security} and Appendix~\ref{app:security-proof}.

\section{Building Blocks of CSTPSI}\label{sec:building-blocks}

We instantiate \mbox{CSTPSI} from three standard primitives:
Shamir secret sharing, the \mbox{BFV} leveled homomorphic encryption
(HE) scheme, and a two-party AES-128 garbled-circuit
\mbox{OPRF}.  We then assemble
\mbox{CSTPSI} from these primitives in
\S\ref{sec:bb-construction} as a self-contained protocol $\Pi_T$,
and state the full protocol in
\S\ref{sec:cstpsi-protocol}.  Notations follow
Table~\ref{tab:notation}.

\subsection{Shamir Secret Sharing}\label{sec:bb-shamir}

A $k$-out-of-$N$ Shamir secret sharing scheme~\cite{shamir1979share}
over a finite field $\mathbb{F}_F$ is a pair of algorithms
$(\mathsf{KS}, \mathsf{KR})$.\footnote{We fix $k = 2$. Raising $k$ inflates the
per-pair work and the RSE bound $\binom{N}{k}\,n_{\mathit{part}}/F^{T}$ through
$\binom{N}{k}$ ($2016$ at $k{=}2$, $41664$ at $k{=}3$) but leaves the per-trial
accept probability at $1/F$; only $T$ lowers that, to $1/F^{T}$
(Theorem~\ref{thm:multi-token}).}  $\mathsf{KS}$ takes a secret
$s \in \mathbb{F}_F$, samples a random degree-$(k-1)$ polynomial
$P(x) = s + a_1 x + \dots + a_{k-1} x^{k-1}$ with
$a_i \getsr \mathbb{F}_F$, and returns the share set
$\{(j, P(j))\}_{j \in [N]}$.  $\mathsf{KR}$ takes any $k$ of these
shares and recovers $s$; fewer
than $k$ shares leave $s$ information-theoretically uniform in
$\mathbb{F}_F$.

\subsection{BFV Homomorphic Encryption}\label{sec:bb-bfv}

The \mbox{BFV} scheme\footnote{We instantiate the \mbox{BFV} scheme 
using Microsoft \mbox{SEAL} 4.1.2~\cite{seal}.} is a semantically
secure leveled homomorphic encryption scheme over the plaintext ring
$\mathbb{Z}_F[X]/(X^m + 1)$, supporting SIMD batching in $m$
parallel slots over $\mathbb{F}_F$ at the parameter regime presented
in Table~\ref{tab:notation}. \mbox{CSTPSI}'s homomorphic evaluation
is a degree-$(s_{\mathit{part}}{-}1)$ univariate polynomial per
query position, evaluated batched across $\lfloor m / N \rfloor$
partitions per ciphertext. 

\subsection{Garbled-Circuit OPRF}\label{sec:bb-gc-oprf}

We realize the OPRF that blinds query and database items as a
two-party AES-128 evaluation in semi-honest Yao garbled
circuits, instantiated via the EMP
toolkit~\cite{wang2017emp}.\footnote{Standard semi-honest 2PC
realization of the AES functionality.}  The sender holds the 128-bit key $\kappa$,
the receiver inputs the query items $(y_1, \dots, y_N)$, and the
receiver alone learns $b_j = \mathrm{AES}_\kappa(y_j) \bmod F$
for $j \in [N]$, with the convention $0 \mapsto 1$ on the zero
pre-image; the sender learns $\bot$.  The sender blinds its own
database locally: $b_{e,i} = \mathrm{AES}_\kappa(x_{e,i})
\bmod F$.

\subsection{Full Protocol Construction}\label{sec:bb-construction}\label{sec:cstpsi-protocol}\label{sec:1-gc}

Figure~\ref{fig:cstpsi-protocol} formally states the CSTPSI
protocol $\Pi_T$ (parametrized by the token-round count $T$)
in full, built from the three primitives above, and realizes
the ideal functionality $\mathcal{F}_{\mathrm{CSTPSI}}$ in
Functionality~\ref{def:F-cstpsi}.
Appendix~\ref{app:protocol-diagram} complements this with a
message-sequence walkthrough (Figure~\ref{fig:cstpsi-sequence})
across the database, sender, and receiver.

\mbox{CSTPSI} runs in an \emph{offline} phase (sender-only) and
an \emph{online} phase (two-party), adopting the standard
\mbox{PSI} / \mbox{HE} optimizations for compressing
communication and reducing HE circuit depth that have become
canonical in the labeled-\mbox{PSI}
line~\cite{chen2018labeled, cong2021labeled, uzun2021fuzzy,
chen2017fast, bgv2014leveled, smartvercauteren2014simd,
psz2015phasing, psww2018cuckoo, fipr2005keyword,
ducasstehle2016sanitization}: \emph{partitioning} the database
into $s_{\mathit{part}}$-sized chunks, baby-step--giant-step
(\emph{BSGS}) windowing of the encrypted query powers,
\emph{SIMD batching} across partitions, and \emph{noise
flooding} plus \emph{modulus switching} on the returned
ciphertexts.  We refer to those works for technique details and
describe only how \mbox{CSTPSI} assembles them below.

We instantiate the design choices motivated in \S\ref
{sec:cstpsi}: per-round share construction (step~\ref{step:share}), 
a single \mbox{AES} garbled circuit reused across all rounds 
(step~\ref{step:gc}), and one cached \mbox{BSGS} query-power window (step~\ref{step:onlinequery}), concretely in Figure~\ref{fig:cstpsi-protocol}.

\smallskip
\textbf{Offline.}  First, $R$ and $S$ agree on the \mbox{BFV}
public parameters (Table~\ref{tab:notation}); $R$ generates the
\mbox{BFV} key pair and shares the public key with $S$.  $S$
then proceeds as follows.  $S$ samples an AES key $\kappa$,
blinds each database item to
$b_{e,i} = \mathrm{AES}_\kappa(x_{e,i}) \bmod F$, and, for
every row $e$ and every round $t \in \{0, \dots, T + K - 1\}$,
draws a fresh degree-$(k{-}1)$ Shamir polynomial
$P_{t,e}(x) = s_{t,e} + a_{t,e}\, x$ with $s_{t,e} = 0$ for
token rounds $t \in [0, T)$ and $s_{t,e}$ equal to the
$(t{-}T{+}1)$-th chunk of $\ell_{x_e}$ for label rounds
$t \in [T, T + K)$.  Per query position $i \in [N]$ and per
partition $p$, $S$ interpolates the unique
degree-$(s_{\mathit{part}}{-}1)$ polynomial $f_{t,i,p}$ over
$\mathbb{F}_F$ through
$\bigl\{\,(b_{e,i},\,P_{t,e}(i)) : e \in p\,\bigr\}$ and
SIMD-packs $\lfloor m / N \rfloor$ partitions per \mbox{SEAL}
plaintext.

\smallskip
\textbf{Online.}  $R$ obtains the blinded query
$(b_1, \dots, b_N) = (\mathrm{AES}_\kappa(y_1), \dots,
\mathrm{AES}_\kappa(y_N)) \bmod F$ via a single garbled-circuit
\mbox{AES} execution, transmits a \mbox{BSGS} sparse window of
encrypted query powers once, and decrypts $f_{t,i,p}(b_j)$ as
$S$ returns the round-$t$ ciphertexts (re-randomized via noise
flooding and shrunk via modulus switching).  Reconstruction
applies the Shamir reconstructor $\mathsf{KR}$ (of
\S\ref{sec:bb-shamir}) to every unordered pair
$(i, j) \in \binom{[N]}{k}$ within each partition; the explicit
$k = 2$ form is given in step~\ref{step:reconstruct-token} of
Figure~\ref{fig:cstpsi-protocol}.  Its Lagrange weights at $x = 0$
depend only on the blinded coordinates $b_i, b_j$, which the single
garbled circuit fixes for the session, so $R$ computes their modular
inverses once per pair and reuses them across all $T + K$ rounds,
sparing a per-round field inversion.  The $T$ token rounds, whose
Shamir secret is the public token $0$, are the independent validators
of \S\ref{sec:cstpsi}; Theorem~\ref{thm:multi-token} bounds their
joint spurious accept at $1/F^{T}$.

\section{Security Analysis}\label{sec:security}

\begin{figure}[t]
  \center
  \resizebox{0.95\width}{!}{
\begin{definition}[Semi-honest realization]\label{def:cstpsi-security}
Let \(\Pi_T\) be a CSTPSI protocol with security parameter
\(\lambda\) and let \(P \in \{S, R\}\) denote a corrupted party.
The \emph{real-world view}
\[
  \Real_{\Pi_T, P}(\lambda; Y, \mathit{Db})
\]
is the tuple \((\View_P, y_R, y_S)\), where \(\View_P\) collects
\(P\)'s input, randomness, and received messages over an honest
execution of \(\Pi_T\) at security parameter \(\lambda\) on inputs
\((Y, \mathit{Db})\), and \((y_R, y_S)\) are the parties' outputs.
The \emph{ideal-world view}
\[
  \Ideal_{\mathcal{F}_{\mathrm{CSTPSI}}, \mathsf{Sim}, P}(\lambda; Y, \mathit{Db})
\]
is obtained by computing \((y_R, y_S) \gets
\mathcal{F}_{\mathrm{CSTPSI}}(Y, \mathit{Db})\) and outputting
\((\mathsf{Sim}(\lambda, P, x_P, y_P),\, y_R, y_S)\), where
\(x_P\) is \(P\)'s input and \(\mathsf{Sim}\) is a probabilistic
polynomial-time simulator.

We say that \(\Pi_T\) \emph{semi-honestly realizes}
\(\mathcal{F}_{\mathrm{CSTPSI}}\) if, for every corrupted party
\(P \in \{S, R\}\), there exists a simulator \(\mathsf{Sim}_P\)
such that for every valid input pair \((Y, \mathit{Db})\),
\[
  \Real_{\Pi_T, P}(\lambda; Y, \mathit{Db})
  \indist
  \Ideal_{\mathcal{F}_{\mathrm{CSTPSI}}, \mathsf{Sim}_P, P}(\lambda; Y, \mathit{Db}).
\]
The maximum advantage of a polynomial-time distinguisher between
the two distributions, taken over both \(P\) and all valid inputs,
is the \emph{simulator gap} of \(\Pi_T\).
\end{definition}
}
\vspace{-0.5em}
\end{figure}

The threat model and the
ideal functionality \(\mathcal{F}_{\mathrm{CSTPSI}}\) are stated in
\S\ref{sec:threat} and Functionality~\ref{def:F-cstpsi}
respectively.  We adopt the standard simulation-based, real/ideal
semi-honest security definition (Definition~\ref{def:cstpsi-security}).  This section states the formal definition, the security
theorem for \(\Pi_T\) of \S\ref{sec:cstpsi-protocol}, and its proof
sketch; the simulator constructions and supporting propositions are
deferred to Appendix~\ref{app:security-proof}.

\subsection{Main Theorem}\label{sec:sec-main}
\begin{figure}[t]
  \center
  \resizebox{0.95\width}{!}{
\begin{theorem}[Semi-honest security of \(\Pi_T\)]\label{thm:main-sec}
    Assuming a semantically secure homomorphic encryption scheme, the
pseudorandomness of \mbox{AES-128}, and a semi-honest-secure garbled-circuit protocol, the
protocol \(\Pi_T\) realizes the ideal functionality
\(\mathcal{F}_{\mathrm{CSTPSI}}\) of Functionality~\ref{def:F-cstpsi}
in the sense of Definition~\ref{def:cstpsi-security} against
semi-honest adversaries with simulator gap at most
\[
\binom{N}{k}\,\frac{n_{\mathit{part}}}{F^{T}}
\;+\; T \cdot \mathsf{Adv}^{\mathrm{prf}}_{\mathrm{AES}}(\mathcal{B})
\;+\; \mathsf{negl}(\lambda),
\]
where the first term is the realization soundness error (RSE)
bound of
Theorem~\ref{thm:multi-token} (via Lemma~\ref{lem:far-bound}), and the
\(\mathsf{negl}(\lambda)\) term absorbs the HE-semantic-security and
\mbox{GC} simulator (privacy) gaps.
\end{theorem}
 }
 \vspace{-1em}
\end{figure}

The simulator gap splits into two parts. The first is a
\emph{soundness} contribution, the RSE mass
\(\binom{N}{k}\,n_{\mathit{part}}/F^{T}\). On a spurious accept the real
protocol deviates from \(\mathcal{F}_{\mathrm{CSTPSI}}\) by accepting a
query that shares fewer than \(k\) items and returning an off-curve
value the ideal functionality would never produce. The second is a
negligible \emph{privacy} contribution, the HE, \mbox{GC}, and
AES-PRF transcript-indistinguishability terms. Because the off-curve
value is independent of the stored records (\S\ref{sec:far-problem}),
it leaks no sender data, so the gap's only operator-controlled term is
a soundness quantity, not a privacy leak.

\begin{proof}[Proof sketch]
Sender privacy (Appendix~\ref{app:sec-sender}): the simulator
\(\mathsf{Sim}_R\) populates non-match reconstructions uniformly in
\(\mathbb{F}_F\); the only event distinguishing it from the real view
is the soundness (RSE) event of Lemma~\ref{lem:far-bound}, bounded by
\(\binom{N}{k}\,n_{\mathit{part}}/F^{T}\) plus the \mbox{AES}-PRF
hybrid cost.  Although Lemma~\ref{lem:far-bound} is stated for a query
that matches no entry, its proof bounds each non-matching entry
independently, so the same union bound covers the non-matching pairs of
an arbitrary query.  Receiver privacy
(Appendix~\ref{app:sec-receiver}): the sender's view comprises one
\mbox{GC} transcript and one ciphertext bundle of encrypted query
powers, both simulable from \((\mathit{Db}, |Y|)\) under the HE
semantic security and the standard semi-honest \mbox{GC} simulator.
Propositions~\ref{prop:multi-round} and~\ref{prop:1-gc}
(Appendix~\ref{app:sec-rounds}) ensure the bound carries across the
\(T+K\) homomorphic rounds and the once-per-session \mbox{GC}
optimization.
\qed
\end{proof}

\paragraph{Scope of the guarantee.}
The simulator gap above is the difference between the real protocol and
the ideal functionality. A natural match, where two unrelated records
share at least $k$ items by chance, is returned by the ideal
functionality as well. It therefore appears in both worlds and cancels
in this gap. What remains is the protocol's RSE on
queries that share fewer than $k$ items, which the bound drives to a
negligible level. The matcher's intrinsic false-match rate (FMR) is a
property of the predicate, not of the protocol, so it is out of scope
here and is the same as in plaintext.

 \vspace{-0.8em}
\subsection{Multi-Token-Round Amplification}\label{sec:sec-amplification}

The bound of Lemma~\ref{lem:far-bound} (§\ref{sec:far-problem}) gives the
single-round $T = 1$ probability that $\Pi_T$ accepts a non-matching
query.  We now show that the amplification
mechanism the construction admits, $T$ independent token rounds with
independently sampled Shamir coefficients per round, drives this bound
exponentially in $T$, with no degradation of the false reject rate.

\begin{figure}[h]
  \center
  \resizebox{0.95\width}{!}{
\begin{theorem}[Multi-Token Amplification]\label{thm:multi-token}
For $T \ge 2$ with independently sampled Shamir coefficients in each token
round and a cryptographically secure source of randomness for those samples,
the RSE of
Lemma~\ref{lem:far-bound} satisfies
\begin{equation}
\Pr[\,\mathrm{accept}\,] \;\le\; \binom{N}{k}\,n_{\mathit{part}}\,\big/\,F^{T}
+ T \cdot \mathsf{Adv}^{\mathrm{prf}}_{\mathrm{AES}}(\mathcal{B}).
\label{eq:far-T-thm}
\end{equation}
Moreover, the false reject rate of $\Pi_T$ matches that of $\Pi_1$:
every true match lies in $\bigcap_{t} S_{t}$ deterministically,
where $S_t$ denotes the set of pairs $(i,j)$ whose round-$t$
reconstruction $\mathsf{KR}(i,j,t)$ equals the public token $0$.
Concrete parameter thresholds are stated in
Corollary~\ref{cor:sufficient-T}.
\end{theorem}
 }
 \vspace{-0.5em}
\end{figure}

\begin{proof}[Proof sketch]
The bound follows immediately from Lemma~\ref{lem:far-bound} once
across-round independence is established.  Independence of the $T$
per-round reconstructions reduces to independence of the Shamir coefficient
samples; in the reference implementation of §\ref{sec:bb-construction},
this requires that each call for creating a secret share draw its
coefficients from a cryptographically secure
source.\footnote{In the reference implementation each thread's
Shamir-coefficient generator is seeded from the operating-system
\mbox{CSPRNG} (\texttt{arc4random\_buf} on macOS/BSD,
\texttt{/dev/urandom} on Linux), replacing the additive-feedback
\texttt{random()} that an earlier prototype used.}  False-reject preservation:
at a true match $(i,j)$ within an enrolled entry $e$, the receiver tries the
same pair $(i,j)$ in every round $t$; the sender's per-entry Shamir
polynomial $P_{t,e}$ is sampled independently per round, but in every round
it has constant term $0$, so the $\mathsf{KR}$ interpolation at $x = 0$
through the points $(b_{e,i}, P_{t,e}(i))$ and $(b_{e,j}, P_{t,e}(j))$
deterministically yields $0$.  Hence $(i,j) \in S_{t}$ for all $t$, and
$(i,j) \in \bigcap_{t} S_{t}$, with no probability of failure.
\qed
\end{proof}

The simulator gap of Theorem~\ref{thm:main-sec} is thus a tunable
parameter, not a fixed weakness.  Its only operator-controlled term,
$\binom{N}{k}\,n_{\mathit{part}}/F^{T}$, is the soundness term of
the composition of spurious accepts; by Corollary~\ref{cor:sufficient-T} an
extra token round drives it below $2^{-\lambda}$ at the target
scale, where it is dominated by the $\mathsf{negl}(\lambda)$ HE and
\mbox{GC} terms and the AES-\mbox{PRF} cost.  A non-negligible gap
at $T = 1$ is therefore an artifact of running below the intended
security parameter, not an inherent limitation.

The sender-privacy simulator construction in
Appendix~\ref{app:sec-sender} uses Theorem~\ref{thm:multi-token} as
its quantitative bound on this soundness term; together with the
receiver-privacy argument and the multi-round propositions of
Appendix~\ref{app:security-proof}, this yields the bound of
Theorem~\ref{thm:main-sec}.

\begin{figure}
  \center
  \resizebox{0.95\width}{!}{
\begin{corollary}[Sufficient Number of Token Rounds]\label{cor:sufficient-T}
For RSE at most $2^{-\lambda}$ with security
parameter $\lambda$, it suffices to choose
\begin{equation}
T \;\ge\; \left\lceil\frac{\lambda \ln 2 + \ln \binom{N}{k}
+ \ln \lceil D/s_{\mathit{part}}\rceil}{\ln F}\right\rceil.
\label{eq:T-suff}
\end{equation}
At the CSTPSI parameters $(N, k, F, s_{\mathit{part}})
= (64, 2, 8\,519\,681, 32)$, $T = 2$ suffices at million-scale
databases under the engineering RSE threshold
($\lambda \approx 20$); $T = 3$ suffices to million-scale at the
cryptographic standard $\lambda = 40$ and to billion-scale under the
engineering threshold; and billion-scale at $\lambda = 40$ needs
$T = 4$.  The \mbox{PRF} advantage term contributes an additive
$T \cdot 2^{-128}$ that does not affect any of these thresholds.
\end{corollary}
}
\vspace{-1em}
\end{figure}

\subsection{Leakage Profile}\label{sec:sec-leak}

\paragraph{Receiver learns} (i) the labels of true matches, as
required by \(\mathcal{F}_{\mathrm{CSTPSI}}\); (ii) the approximate
database size via the number of \mbox{SIMD} partitions transmitted
(low severity; this parameter may be public); (iii) at \(T = 1\),
spurious accepts that return off-curve values at the rate of
Lemma~\ref{lem:far-bound}, the soundness term of
Theorem~\ref{thm:main-sec} (driven to negligible at \(T \ge 2\)); and (iv) on each
matched partition, the count and the identities of the agreeing
component positions (revealed only when the count reaches \(k\)
and a label is recovered).

The component-position leakage of (iv) is inherent to the
set-threshold reduction and is shared by every protocol in this
class; what distinguishes \mbox{CSTPSI} is that the multi-token
amplification confines it to genuine matches.  At \(T = 1\) the same
indices leak on each spurious accept, that is, for a
\(\binom{N}{k}\,n_{\mathit{part}}/F\) fraction of non-matching rows;
driving that rate to negligible at \(T \ge 2\) is precisely what
restricts the leakage to true matches.

\paragraph{Sender learns} nothing of cryptographic value.  The
constant \mbox{GC} online cost eliminates the
label-size timing side channel that a naive multi-round-\mbox{GC}
composition would exhibit.

\paragraph{Network observer learns} the round count \(T + K\),
from which an observer can infer the chosen label size unless that
parameter is treated as public, and fixed-size message bodies that
carry no item-specific information.

\section{Evaluation}\label{sec:eval}

\begin{table*}[t]
\centering
\caption{Per-query online time (s) and communication (MiB), grouped by
thread count (1, 4, 8). CSTPSI runs at $T = 2$ token rounds, STLPSI at
$T = 1$. Faster and lower values are bolded; (-) marks cases where CSTPSI
sends more than the baseline.}
\label{tab:headtohead}
\small
\setlength{\tabcolsep}{4pt}
\renewcommand{\arraystretch}{1.0}

\begin{tabular}{ll ccc ccc ccc ccc}
\toprule
& & \multicolumn{3}{c}{\textbf{Online Time, 1\,thr (s)}} & \multicolumn{3}{c}{\textbf{Online Time, 4\,thr (s)}} & \multicolumn{3}{c}{\textbf{Online Time, 8\,thr (s)}} & \multicolumn{3}{c}{\textbf{Comm. (MiB)}} \\
\cmidrule(lr){3-5} \cmidrule(lr){6-8} \cmidrule(lr){9-11} \cmidrule(lr){12-14}
$D$ & Label & STL & CST & Spd. & STL & CST & Spd. & STL & CST & Spd. & STL & CST & Save \\
\midrule

1K
& 23-bit & 3.3 & \textbf{1.7} & 1.9$\times$ & 3.3 & \textbf{1.7} & 2.0$\times$ & 3.4 & \textbf{1.7} & 2.0$\times$ & 6 & \textbf{3} & 50\% \\
& 16-B   & 11.3 & \textbf{1.7} & 6.5$\times$ & 11.5 & \textbf{1.7} & 6.7$\times$ & 11.4 & \textbf{1.7} & 6.7$\times$ & 20 & \textbf{3} & 85\% \\
& 32-B   & 21.3 & \textbf{1.9} & 11.5$\times$ & 21.0 & \textbf{1.7} & 12.0$\times$ & 21.1 & \textbf{1.8} & 12.0$\times$ & 36 & \textbf{4} & 89\% \\
& 64-B   & 38.8 & \textbf{2.0} & 19.4$\times$ & 40.1 & \textbf{1.9} & 21.2$\times$ & 40.1 & \textbf{1.9} & 21.2$\times$ & 67 & \textbf{5} & 93\% \\
\midrule

10K
& 23-bit & 3.4 & \textbf{1.9} & 1.8$\times$ & 3.3 & \textbf{1.7} & 1.9$\times$ & 3.5 & \textbf{1.8} & 1.9$\times$ & 6 & \textbf{4} & 33\% \\
& 16-B   & 11.8 & \textbf{2.2} & 5.3$\times$ & 11.8 & \textbf{1.9} & 6.2$\times$ & 11.5 & \textbf{1.8} & 6.4$\times$ & 22 & \textbf{6} & 73\% \\
& 32-B   & 21.9 & \textbf{2.7} & 8.1$\times$ & 21.9 & \textbf{2.1} & 10.6$\times$ & 21.5 & \textbf{1.9} & 11.1$\times$ & 41 & \textbf{9} & 78\% \\
& 64-B   & 40.6 & \textbf{3.5} & 11.5$\times$ & 41.2 & \textbf{2.4} & 17.3$\times$ & 39.7 & \textbf{2.2} & 18.4$\times$ & 75 & \textbf{14} & 81\% \\
\midrule

100K
& 23-bit & 4.9 & \textbf{4.0} & 1.2$\times$ & 3.9 & \textbf{2.4} & 1.6$\times$ & 3.8 & \textbf{2.3} & 1.7$\times$ & \textbf{14} & 16 & -  \\
& 16-B   & 16.6 & \textbf{7.7} & 2.2$\times$ & 12.9 & \textbf{3.5} & 3.7$\times$ & 12.7 & \textbf{3.0} & 4.2$\times$ & 49 & \textbf{37} & 24\% \\
& 32-B   & 31.4 & \textbf{12.1} & 2.6$\times$ & 24.8 & \textbf{4.9} & 5.0$\times$ & 23.2 & \textbf{3.9} & 5.9$\times$ & 91 & \textbf{63} & 31\% \\
& 64-B   & 56.7 & \textbf{20.0} & 2.8$\times$ & 44.7 & \textbf{7.4} & 6.0$\times$ & 45.6 & \textbf{5.9} & 7.8$\times$ & 169 & \textbf{111} & 34\% \\
\midrule

1M
& 23-bit & \textbf{19.6} & 25.9 & 0.8$\times$ & \textbf{8.3} & 9.1 & 0.9$\times$ & \textbf{6.7} & 6.8 & 1.0$\times$ & \textbf{92} & 132 & - \\
& 16-B   & 65.3 & \textbf{62.7} & 1.0$\times$ & 28.3 & \textbf{20.5} & 1.4$\times$ & 23.0 & \textbf{14.6} & 1.6$\times$ & \textbf{322} & 349 & - \\
& 32-B   & 119.9 & \textbf{106.9} & 1.1$\times$ & 52.1 & \textbf{34.1} & 1.5$\times$ & 41.9 & \textbf{23.9} & 1.8$\times$ & \textbf{597} & 608 & -\\
& 64-B   & 219.7 & \textbf{187.8} & 1.2$\times$ & 95.6 & \textbf{59.1} & 1.6$\times$ & 79.9 & \textbf{42.2} & 1.9$\times$ & 1103 & \textbf{1084} & 2\% \\
\bottomrule
\end{tabular}
\vspace{-1.5em}
\end{table*}

\subsection{Experimental Setup}\label{sec:eval-setup}

We implement CSTPSI in C++ from scratch; the source and a one-command
reproduction harness are publicly available under the PolyForm
Noncommercial License.\footnote{\url{https://github.com/euzun/cstpsi}}
All experiments run on a single workstation with a 12-core
(6 performance, 6 efficiency) processor and $36$\,GB of memory under a
current desktop operating system.
The sender and receiver run as separate processes communicating over
the loopback interface with \mbox{ZeroMQ}~\cite{zeromq}; network
conditions are simulated at $10$\,Gbps bandwidth and $0.02$\,ms
latency.
Each case uses $300$ true-positive and $300$ true-negative query draws,
reduced to $100$ of each for the heaviest configurations ($32$- and
$64$-byte labels and all $D = 1$M cases); timing, communication, and
memory are reported as the median over independent runs.
The cryptographic stack is \mbox{SEAL}~\cite{seal} (BFV at
$\approx\!128$-bit security) for the homomorphic layer,
\mbox{EMP-toolkit}~\cite{wang2017emp}
for the AES-128 oblivious-PRF garbled circuit, and \mbox{FLINT}~\cite{flint}
with GMP and MPFR for polynomial arithmetic; threading uses OpenMP.

We use two data regimes. The synthetic regime
(Sections~\ref{sec:eval-far} and~\ref{sec:eval-perf}) realizes a
perfect matcher with zero false-match rate. Every true-negative query
is built to share fewer than $k$ items with every enrolled record, so
no natural match can occur and every accept is purely the kernel's
realization soundness error. This isolates the RSE and shows clearly
how it scales with the database size. The real regime
(Section~\ref{sec:eval-real}) uses two public datasets at different
scales, LFW~\cite{huang2008labeled} and Deep1B~\cite{babenko2016efficient}, to demonstrate the attack on a real matcher and
to show that CSTPSI restores the RSE to zero.

\subsection{Benchmarked Protocols and Headline Result}\label{sec:eval-protocols}

We also implemented the STLPSI protocol of Uzun et
al.~\cite{uzun2021fuzzy} from scratch and exposed it within the same
tool as an alternative mode, so that both protocols run on identical
code paths and hardware. The two modes are:
\begin{itemize}
  \item \textbf{CSTPSI}: the deployment configuration, which evaluates
        the garbled circuit once per session, transmits the encrypted
        query once, and applies the multi-token-round mitigation
        that drives the realization soundness error (RSE) to its
        amplified bound.
  \item \textbf{STLPSI}: the baseline, with a single token round; the oblivious-PRF garbled circuit and the encrypted
        query are recomputed and retransmitted on every protocol round.
\end{itemize}

Both modes share the standard labeled-PSI optimizations
(partitioning, \mbox{SIMD} batching, baby-step--giant-step
query-power windowing), differing only in whether the garbled
circuit and encrypted query are cached across rounds (CSTPSI) or
recomputed each round (STLPSI); the reported speedups thus reflect
this protocol-level caching, not an unoptimized baseline.

The protocol parameters and the swept ranges are collected in
Table~\ref{tab:notation} (\S\ref{sec:overview}); the security level is
$\approx\!128$ bits.

\paragraph{Headline result}
At a million records ($D = 1$M, $8$ threads) the no-caching STLPSI
baseline ($T = 1$) has an RSE of one, admitting a spurious accept on
every spoofing query, while CSTPSI ($T = 2$) admits none. CSTPSI buys
this soundness essentially for free. Across the tested range it matches
or beats the baseline, running more than $20\times$ faster and sending
up to $93\%$ less data at best (at $D = 1$K with a $64$-byte label) and
staying within about $2\%$ of it at worst.
Corollary~\ref{cor:sufficient-T} projects an engineering-negligible RSE
(${\leq}10^{-6}$) through a million records, which the run confirms.

\begin{figure}[!h]
  \centering
  \vspace{-0.8em}
  \begin{tikzpicture}[x=1cm,y=1cm]
    \foreach \xx/\xl in {0/{$10^3$},2/{$10^4$},4/{$10^5$},6/{$10^6$}}{
      \draw[gray!18] (\xx,0)--(\xx,4); \node[font=\scriptsize,below] at (\xx,0){\xl};}
    \foreach \yy/\yl in {0/{$10^{-3}$},1.333/{$10^{-2}$},2.667/{$10^{-1}$},4/{$1$}}{
      \draw[gray!18] (0,\yy)--(6,\yy); \node[font=\scriptsize,left] at (0,\yy){\yl};}
    \draw[black!60] (0,0) rectangle (6,4);
    \draw[cLem,very thick] plot[smooth] coordinates
      {(0,1.165)(0.954,1.796)(2,2.479)(2.954,3.072)(4,3.630)(4.954,3.935)(6,4)};
    \draw[gray,thick,dashed] (0,1.167)--(4.25,4);
    \foreach \px/\py in {2/2.204,4/3.667,6/3.96}{\fill[cThm] (\px,\py) circle (2.2pt);}
    \node[font=\small,below=5mm] at (3,0){Database size $D$};
    \node[font=\small,rotate=90] at (-1.05,2){RSE ($T{=}1$)};
    \begin{scope}[shift={(0.3,3.5)},font=\scriptsize]
      \draw[cLem,very thick] (0,0)--(0.45,0); \node[right] at (0.45,0){exact bound};
      \draw[gray,thick,dashed] (0,-0.38)--(0.45,-0.38); \node[right] at (0.45,-0.38){linear approx.};
      \fill[cThm] (0.225,-0.76) circle (2.2pt); \node[right] at (0.45,-0.76){measured};
    \end{scope}
    \node[font=\scriptsize,align=center] at (4.25,0.7)
      {$T{\geq}2$: no spurious accept\\(bound $10^{-10}$ to $10^{-7}$)};
  \end{tikzpicture}
  \vspace{-0.5em}
  \caption{Empirical RSE versus database size $D$ (log--log axes).}
  \label{fig:far-empirical}
  \vspace{-1.7em}
\end{figure}
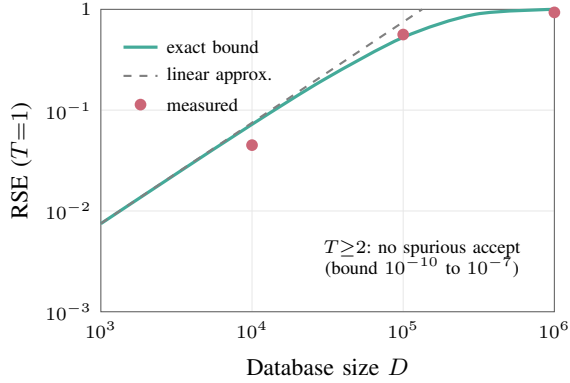

\subsection{Correctness and Realization Soundness Error}\label{sec:eval-far}

We measure the protocol's RSE, not the matcher FMR. Every true-negative
query is drawn to share fewer than $k$ items with every enrolled record,
so any accept is a protocol artifact and not a natural match; the
intrinsic FMR is a plaintext property of the representation, reported as
the FMR floor in \S\ref{sec:eval-real}.

Fig.~\ref{fig:far-empirical} plots the measured $T = 1$ RSE against the
bound $\binom{N}{k}\,n_{\mathit{part}}/F^{T}$ of Lemma~\ref{lem:far-bound}.
It grows with $D$ as predicted, from under $1\%$ at $1$K through $4.5\%$ at
$10$K and $56\%$ at $100$K to $100\%$ at $1$M, tracking the exact bound;
the linear approximation overestimates above $D \approx 10^{4}$. The $1$K
point is omitted from the log plot, as no spurious accept was observed in
its sample (consistent with the ${\sim}1\%$ bound). A single token round
removes the error: no spurious accept was observed at $T \geq 2$ through
$D = 1$M, where the bound falls from $10^{-10}$ at $1$K to $10^{-7}$ at
$1$M, and Corollary~\ref{cor:sufficient-T} gives the sufficient $T$ for any
target scale. Every true-positive query is accepted (false-reject rate $0$
throughout), so the fix costs no recall.

\subsection{Soundness on Real Datasets}\label{sec:eval-real}

The synthetic study isolates the kernel's RSE; we now confirm the same
effect on real data at two scales. Deep1B carries the scale story up to
a million records, and LFW checks end-to-end accuracy on a real
biometric matcher. We measure the per-query RSE as the fraction
of queries with at least one spurious accept. Such an accept returns an
off-curve value, not an enrolled label (\S\ref{sec:far-problem}), so this
fraction measures the composition gap directly.

\paragraph{Datasets and setup}
Deep1B~\cite{babenko2016efficient} includes 96-dimensional deep object
descriptors. We enroll $D$ records up to $10^{6}$ and issue, per repeat,
$100$ genuine queries (each close to one enrolled record and in the published query vectors) and $100$
random queries (no planted similarity), over $10$ repeats.
LFW~\cite{huang2008labeled} is a widely used dataset for face-recognition accuracy. It supplies $12{,}209$ face embeddings; we use the $128$-dimensional vectors. We built a lightweight FLPSI wrapper that consumes these embedding vectors and reduces them to $N$-item sets, with the exact parameters and conversion flow prescribed in FLPSI'21~\cite{uzun2021fuzzy}. The wrapper and the input embeddings for this experiment are in our released codebase.
We enroll $D = 1500$ records and issue, per repeat, a genuine query for
each enrolled record and $100$ impostor queries, over $10$ repeats. Both
datasets pass through the same locality-sensitive hash to a $256$-bit
sketch and use $T = 2$ for CSTPSI.

\begin{table}[h]
\centering
\caption{Deep1B kernel soundness vs.\ DB size: the baseline RSE
climbs to one. CSTPSI stays at zero (FRR identical).}
\label{tab:deep1b}
\small
\setlength{\tabcolsep}{6pt}
\begin{tabular}{l cccc}
\toprule
$D$ & $10^{3}$ & $10^{4}$ & $10^{5}$ & $10^{6}$ \\
\midrule
$k$-of-$N$ FRR        & 0.112 & 0.112 & 0.112 & 0.112 \\
STLPSI RSE ($T{=}1$)  & 0.010 & 0.080 & 0.504 & 0.999 \\
CSTPSI RSE ($T{=}2$)  & \textbf{0.000} & \textbf{0.000} & \textbf{0.000} & \textbf{0.000} \\
\bottomrule
\end{tabular}
\vspace{-0.6em}
\end{table}

\paragraph{Deep1B: the RSE approaches one}
The baseline (STLPSI, $T = 1$) RSE climbs with $D$, from $0.010$
at $10^{3}$ to $0.080$, $0.504$, and $0.999$ at $10^{6}$, matching the
synthetic curve of Fig.~\ref{fig:far-empirical}; at a million records
almost every query produces a spurious accept, about seven off-curve
values per query.\footnote{The split counts a reconstruction as a label
when its value falls in $[0,D)$ and as off-curve otherwise. An off-curve
value falls in $[0,D)$ by chance with probability $D/F$, which is
$11.7\%$ at $D = 10^{6}$, so the reported count is a lower bound on the
true off-curve rate.} CSTPSI holds the RSE at zero throughout,
and the false-reject rate is identical for the plaintext matcher,
STLPSI, and CSTPSI ($0.112$), so the fix costs no recall. The random and
genuine queries behave alike, confirming the effect follows from the
database size, not the data, as Lemma~\ref{lem:far-bound} predicts.

The matcher's own false matches run far higher. Its FMR floor rises from
about $0.95$ to $267$ real matches per query as $D$ grows from $10^{3}$
to $10^{6}$, but that floor is a plaintext property of the
representation, shared by the matcher, STLPSI, and CSTPSI, and out of
scope here (\S\ref{sec:threat-errors}). What the protocol owns is the
RSE, which STLPSI drives to one and CSTPSI removes entirely. The floor
being the larger number is exactly why the two must be separated. CSTPSI
zeros the protocol's contribution, bringing the deployed system down to
the matcher's floor, which amplification cannot and should not go below
(\S\ref{sec:discussion}).

\begin{table}[t]
\centering
\caption{LFW end-to-end accuracy at the plaintext $k$-of-$N$ operating
point. CSTPSI matches the matcher's false-accept and false-reject rates;
STLPSI inflates the false-accept rate.}
\label{tab:lfw}
\small
\begin{tabular}{l cc}
\toprule
Configuration & FAR & FRR \\
\midrule
$k$-of-$N$ (plaintext)  & $0.337$ & $0.324$ \\
STLPSI ($T{=}1$)        & $0.352$ & $0.324$ \\
CSTPSI ($T{=}2$)        & $0.337$ & $0.324$ \\
\bottomrule
\end{tabular}
\vspace{-1.2em}
\end{table}

\paragraph{LFW: a real biometric matcher}
Table~\ref{tab:lfw} reports end-to-end behavior on faces. LFW is a
stringent test because the plaintext matcher is noisy at this operating
point: at the $k$-of-$N$ threshold it already sits at a high false-accept
rate of $0.337$ and false-reject rate of $0.324$. The test is
faithfulness, not accuracy.
Whatever the plaintext operating point, the kernel must reproduce it.
CSTPSI does: it matches the $k$-of-$N$ false-accept and false-reject
rates ($0.337$ and $0.324$) to the digit, so the cryptographic layer adds
nothing to the matcher's own error. STLPSI, without the soundness fix,
instead pushes the false-accept rate above the matcher floor, to $0.352$.
Deep1B carries the scale story; LFW shows the kernel tracks a real
matcher's operating point under encryption, however noisy that matcher
is.

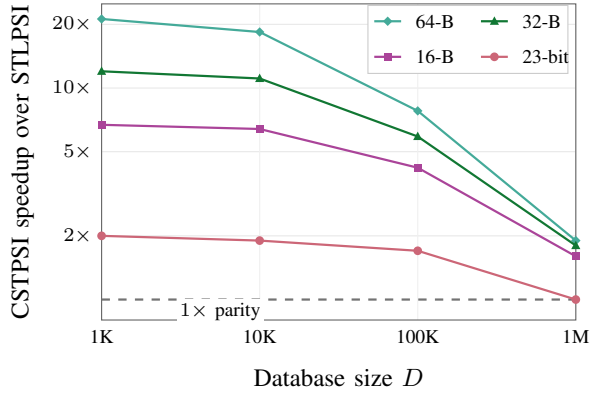
\begin{figure}[!h]
  \centering
  \vspace{-0.5em}
  \resizebox{0.92\columnwidth}{!}{
  \begin{tikzpicture}[x=1cm,y=1cm]
    \foreach \yy/\yl in {1.065/{$2\times$},2.130/{$5\times$},2.935/{$10\times$},3.741/{$20\times$}}{
      \draw[gray!18] (0,\yy)--(6,\yy); \node[font=\scriptsize,left] at (0,\yy){\yl};}
    \foreach \xx/\xl in {0/{1K},2/{10K},4/{100K},6/{1M}}{
      \draw[gray!12] (\xx,0)--(\xx,4); \node[font=\scriptsize,below] at (\xx,0){\xl};}
    \draw[black!60] (0,0) rectangle (6,4);
    \draw[black!55,dashed,thick] (0,0.259)--(6,0.259);
    \node[font=\scriptsize,fill=white,inner sep=1pt] at (1.5,0.11){$1\times$ parity};
    \draw[cLem,thick] (0,3.809)--(2,3.644)--(4,2.646)--(6,1.005);
    \foreach \x/\y in {0/3.809,2/3.644,4/2.646,6/1.005}{
      \fill[cLem] (\x,\y+0.06)--(\x-0.06,\y)--(\x,\y-0.06)--(\x+0.06,\y)--cycle;}
    \draw[cFunc,thick] (0,3.147)--(2,3.056)--(4,2.322)--(6,0.942);
    \foreach \x/\y in {0/3.147,2/3.056,4/2.322,6/0.942}{
      \fill[cFunc] (\x,\y+0.062)--(\x-0.058,\y-0.045)--(\x+0.058,\y-0.045)--cycle;}
    \draw[cCor,thick] (0,2.470)--(2,2.417)--(4,1.927)--(6,0.806);
    \foreach \x/\y in {0/2.470,2/2.417,4/1.927,6/0.806}{
      \fill[cCor] (\x-0.05,\y-0.05) rectangle (\x+0.05,\y+0.05);}
    \draw[cThm,thick] (0,1.065)--(2,1.005)--(4,0.876)--(6,0.259);
    \foreach \x/\y in {0/1.065,2/1.005,4/0.876,6/0.259}{
      \fill[cThm] (\x,\y) circle (0.055);}
    \node[font=\small,below=5mm] at (3,0){Database size $D$};
    \node[font=\small,rotate=90] at (-1.0,2){CSTPSI speedup over STLPSI};
    \begin{scope}[shift={(3.45,3.78)},font=\scriptsize]
      \fill[white] (-0.08,0.2) rectangle (2.45,-0.66);
      \draw[black!20] (-0.08,0.2) rectangle (2.45,-0.66);
      \draw[cLem,thick](0,0)--(0.34,0);
        \fill[cLem](0.17,0.06)--(0.11,0)--(0.17,-0.06)--(0.23,0)--cycle;
        \node[right] at (0.40,0){64-B};
      \draw[cFunc,thick](1.35,0)--(1.69,0);
        \fill[cFunc](1.52,0.062)--(1.462,-0.045)--(1.578,-0.045)--cycle;
        \node[right] at (1.75,0){32-B};
      \draw[cCor,thick](0,-0.45)--(0.34,-0.45);
        \fill[cCor](0.12,-0.5) rectangle (0.22,-0.4);
        \node[right] at (0.40,-0.45){16-B};
      \draw[cThm,thick](1.35,-0.45)--(1.69,-0.45);
        \fill[cThm](1.52,-0.45) circle (0.055);
        \node[right] at (1.75,-0.45){23-bit};
    \end{scope}
  \end{tikzpicture}
  }
  \caption{Per-query online-time speedup of CSTPSI over STLPSI ($8$
    threads) versus database size $D$, for each label size (both axes
    log scale; the dashed line marks $1\times$, parity).}
  \label{fig:perf-speedup}
  \vspace{-1.2em}
\end{figure}
\subsection{Performance of CSTPSI}\label{sec:eval-perf}

Having established that the deployment configuration zeroes the
RSE, we show it does so without sacrificing performance.
Table~\ref{tab:headtohead} reports the head-to-head of CSTPSI against the
STLPSI baseline across database sizes, label sizes, and thread counts,
with the per-case speedup and communication saving.

\textbf{Scaling with database size.}
CSTPSI's edge over the baseline is largest on small and moderate
databases (Fig.~\ref{fig:perf-speedup}). Its homomorphic work per round
grows with $D$, so as the database grows this per-round cost catches up
with the one-time saving from running GC once. Through $D = 100$K CSTPSI is
faster at every label size. At $D = 1$M it is still faster at all but the
smallest label: with a 23-bit label the extra security round of
T=2 outweighs GC saving, and the two protocols run
within about $2\%$ of each other. CSTPSI thus matches or beats the baseline
at every operating point we measure, with no regime in which it is more
than marginally slower.

\textbf{Scaling with label size.}
STLPSI re-runs GC and retransmits the encrypted
query on every round, where they grow with the label. CSTPSI
runs GC once and sends the query once, so its advantage widens
as labels grow (the upper curves in Fig.~\ref{fig:perf-speedup}). The
effect peaks at $D = 1$K with a $64$-byte label, where CSTPSI is more
than $20\times$ faster than the baseline at $8$ threads (Fig.~\ref{fig:perf-speedup})
and sends about $93\%$ less data.

\textbf{Scaling with threads.}
Threads help most where homomorphic evaluation dominates, that is at
large databases. At $D = 100$K the speedup from $1$ to $8$ threads
reaches about $2.7\times$. At $D = 1$K the garbled circuit runs
single-threaded, leaving little to parallelize, so extra threads
change little.

\textbf{Communication.}
Per-query communication has two parts, the receiver's query upload
($R{\to}S$) and the sender's encrypted response ($S{\to}R$);
Table~\ref{tab:headtohead} and Fig.~\ref{fig:comm-saving} give the
per-case breakdown. The key effect is that caching the encrypted query
fixes $R{\to}S$ at a constant $2.7$\,MiB per query in CSTPSI,
independent of database and label size, whereas STLPSI re-sends it on
each of its $1 + K$ rounds; the response $S{\to}R$ instead scales with
$D$ in both protocols and dominates at scale. Caching therefore wins
decisively where the upload dominates, on small to moderate databases
with large labels (up to $93\%$ less data), and the advantage narrows or
reverses at the largest databases, where the extra $T = 2$ response
round makes CSTPSI send about $40\%$ more for a $23$-bit label at
$D = 1$M. These extra bytes are modest in absolute terms and are the
price of the security round whose benefit Section~\ref{sec:eval-far}
quantifies.

\textbf{Peak RAM.}
Peak resident memory is dominated by the sender, which holds the
partitioned database, and grows with $D$. It rises from well under
$1$\,GB at small databases to about $5.3$--$5.6$\,GB at $D = 100$K
($64$-byte) and $8.0$--$9.1$\,GB at $D = 1$M ($23$-bit, $T = 2$), a
modest overhead over the baseline; the receiver's footprint stays under
$1$\,GB throughout.

\begin{figure}[!h]
  \centering
  \vspace{-1.6em}
  \resizebox{0.92\columnwidth}{!}{
  \begin{tikzpicture}[x=1cm,y=1cm]
    \foreach \yy/\yl in {0/{$-50\%$},2.0/{$25\%$},2.667/{$50\%$},3.333/{$75\%$},4/{$100\%$}}{
      \draw[gray!18] (0,\yy)--(6,\yy); \node[font=\scriptsize,left] at (0,\yy){\yl};}
    \foreach \xx/\xl in {0/{1K},2/{10K},4/{100K},6/{1M}}{
      \draw[gray!12] (\xx,0)--(\xx,4); \node[font=\scriptsize,below] at (\xx,0){\xl};}
    \draw[black!60] (0,0) rectangle (6,4);
    \draw[black!55,dashed,thick] (0,1.333)--(6,1.333);
    \node[font=\scriptsize,fill=white,inner sep=1pt] at (1.1,1.13){$0\%$ parity};
    \draw[cLem,thick] (0,3.800)--(2,3.501)--(4,2.248)--(6,1.379);
    \foreach \x/\y in {0/3.800,2/3.501,4/2.248,6/1.379}{
      \fill[cLem] (\x,\y+0.06)--(\x-0.06,\y)--(\x,\y-0.06)--(\x+0.06,\y)--cycle;}
    \draw[cFunc,thick] (0,3.704)--(2,3.413)--(4,2.155)--(6,1.285);
    \foreach \x/\y in {0/3.704,2/3.413,4/2.155,6/1.285}{
      \fill[cFunc] (\x,\y+0.062)--(\x-0.058,\y-0.045)--(\x+0.058,\y-0.045)--cycle;}
    \draw[cCor,thick] (0,3.600)--(2,3.272)--(4,1.987)--(6,1.109);
    \foreach \x/\y in {0/3.600,2/3.272,4/1.987,6/1.109}{
      \fill[cCor] (\x-0.05,\y-0.05) rectangle (\x+0.05,\y+0.05);}
    \draw[cThm,thick] (0,2.667)--(2,2.213)--(4,0.952)--(6,0.173);
    \foreach \x/\y in {0/2.667,2/2.213,4/0.952,6/0.173}{
      \fill[cThm] (\x,\y) circle (0.055);}
    \node[font=\small,below=5mm] at (3,0){Database size $D$};
    \node[font=\small,rotate=90] at (-1.05,2){CSTPSI comm.\ saving over STLPSI};
    \begin{scope}[shift={(3.45,3.78)},font=\scriptsize]
      \fill[white] (-0.08,0.2) rectangle (2.45,-0.66);
      \draw[black!20] (-0.08,0.2) rectangle (2.45,-0.66);
      \draw[cLem,thick](0,0)--(0.34,0);
        \fill[cLem](0.17,0.06)--(0.11,0)--(0.17,-0.06)--(0.23,0)--cycle;
        \node[right] at (0.40,0){64-B};
      \draw[cFunc,thick](1.35,0)--(1.69,0);
        \fill[cFunc](1.52,0.062)--(1.462,-0.045)--(1.578,-0.045)--cycle;
        \node[right] at (1.75,0){32-B};
      \draw[cCor,thick](0,-0.45)--(0.34,-0.45);
        \fill[cCor](0.12,-0.5) rectangle (0.22,-0.4);
        \node[right] at (0.40,-0.45){16-B};
      \draw[cThm,thick](1.35,-0.45)--(1.69,-0.45);
        \fill[cThm](1.52,-0.45) circle (0.055);
        \node[right] at (1.75,-0.45){23-bit};
    \end{scope}
  \end{tikzpicture}
  }
  \caption{Communication saving of CSTPSI over STLPSI versus database
    size $D$, for each label size ($0\%$ is parity; negative values
    mean CSTPSI sends more).}
  \label{fig:comm-saving}
  \vspace{-1.4em}
\end{figure}
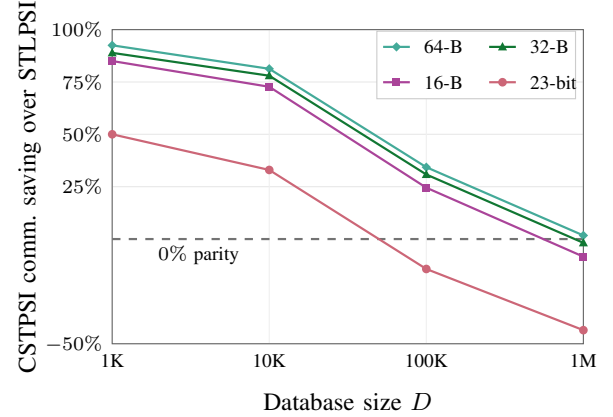

\subsection{Optimization Breakdown}\label{sec:eval-ablation}
We add CSTPSI's caching optimizations one at a time and report each step
as a factor on the per-query online time, with no absolute timings, so
the contribution reads independently of any single case. The base is
CSTPSI at $T = 2$ with no caching, which is exactly the baseline run at
$T = 2$. Relative to the original $T = 1$ baseline this base pays only
the security fix, one extra homomorphic round out of $T + K$, a tax of
about $10\%$ at the $23$-bit label and $3\%$ at the $64$-byte label. On
top of the base we add three caching steps in sequence
(Table~\ref{tab:ablation}): a single garbled circuit per session, with
the receiver still re-sending its encrypted query powers each round;
then the sender caching the expanded power window, so the receiver sends
it only once; then the receiver caching its reconstruction inverses.

\begin{table}[h]
\centering
\caption{Cumulative speedup of each caching step over the unoptimized
base (CSTPSI at $T{=}2$, no caching), per database size (factors).}
\label{tab:ablation}
\small
\setlength{\tabcolsep}{3pt}
\begin{tabular}{l cccc}
\toprule
& \multicolumn{4}{c}{Speedup over base at $D=$} \\
\cmidrule(lr){2-5}
Caching step (cumulative) & $1$K & $10$K & $100$K & $1$M \\
\midrule
$+$ once-per-session GC            & $21.44\times$ & $18.59\times$ & $7.91\times$ & $1.95\times$ \\
$+$ sender power-window cache       & $22.06\times$ & $19.06\times$ & $7.98\times$ & $1.95\times$ \\
$+$ receiver reconstruction cache  & $22.07\times$ & $19.13\times$ & $8.09\times$ & $1.99\times$ \\
\bottomrule
\end{tabular}
\vspace{-0.6em}
\end{table}

Every entry in Table~\ref{tab:ablation} is a speedup factor over the
no-caching base. The garbled-circuit step dominates. Evaluating it once
instead of on each of the $1 + K$ rounds removes a $K/(1+K)$ fraction of
the oblivious-PRF cost, from $50\%$ at the $23$-bit label to $96\%$ at
the $64$-byte label. Because that cost dominates the baseline at small
and moderate $D$, this single step alone recovers the security fix's tax
many times over. The sender power-window cache then uploads the
encrypted query powers once rather than on every round. The
reconstruction cache replaces a per-pair field inversion with a cached
lookup, which makes the receiver reconstruction about $2.5\times$
faster. That saving grows with $D$, from under a millisecond at $1$K to
about $0.8$ seconds at $1$M, yet it stays a small share of the online
time, which the homomorphic evaluation dominates. End to end, the
optimizations turn the added $T = 2$ round from a cost into a net
speedup over the baseline across the tested range.

\subsection{Artifact and Reproducibility}\label{sec:eval-artifact}

We provide our complete source together with single-command shell
scripts that regenerate every table and figure in this section from
fixed seeds, and an interactive tool that runs the protocol
end to end. The artifact and documentation are
publicly available at \url{https://github.com/euzun/cstpsi}.

\section{Discussion}\label{sec:discussion}
\textbf{Discussion.} CSTPSI's central property is that caching decouples the dominant online costs from the label size, so its advantage over the baseline widens as labels grow (\S\ref{sec:eval-perf}). One deployment nuance is worth noting. On a fast ($10$ Gbps) link the homomorphic evaluation, not transmission, dominates latency, so the caching primarily reduces communication volume rather than wall-clock time. Throughout, two token rounds keep the realization soundness error (RSE) below engineering thresholds for databases up to one million records.

\textbf{Generalization.} The composition gap is not specific to this kernel. Any realization that buys efficiency with a per-trial false accept and runs one trial per record inherits it, and the set-threshold kernel is simply where the gap separates cleanly from the matcher's FMR floor and can be measured (\S\ref{sec:threat}). The remedy generalizes as a design principle. When a matcher's soundness rests on a per-trial event whose mass scales with the workload, make the check a first-class, independent, in-protocol validator rather than an after-the-fact output filter, and size the number of validators to the deployment scale through a closed-form bound (Lemma~\ref{lem:far-bound}); soundness then composes with the workload instead of degrading with it, since a spurious accept must fool all $T$ validators at once. CSTPSI is the threshold-labeled fuzzy-PSI instance, and we state this as a measured observation about a class of protocols, not a proof of the general case.

\textbf{Limitations.} CSTPSI removes the kernel's RSE, but it does not change the matcher's intrinsic FMR. Two unrelated records can still share at least $k$ items by chance, and the expected number of such natural matches per query grows with the database size. The application controls this rate through the representation and the choice of $N$ and $k$, not through the protocol. CSTPSI therefore brings the deployed system down to this FMR floor, and amplification cannot go below it. Amplification is free on the false-reject side, because a true match shares the actual items and reconstructs in every round, so adding token rounds removes only spurious accepts. Reducing the FMR itself would instead need independent encodings across the rounds, which would raise the false-reject rate, and we leave this to future work. A second limit is communication. The sender-to-receiver response grows by one field element per round, so it scales with the label length, and the construction targets small-to-moderate labels rather than bulk payloads (at one field element of about $2.8$ bytes per round, a one-megabyte label would need roughly $350{,}000$ rounds). A third limit is the threshold, fixed at $k = 2$. A larger $k$ inflates the factor $\binom{N}{k}$ (for example, from $\binom{64}{3} = 41{,}664$ to $\binom{64}{5} \approx 7.6 \times 10^{6}$) and raises the single-round RSE, so it needs more token rounds for the same target. Seen positively, amplification is what makes a larger $k$ usable at all, because without it the single-round rate at $k = 3$ already approaches certainty on modest databases. Finally, our analysis is semi-honest, and malicious security is future work. Amplification raises the cost of cryptographic probing, where an adversary crafts queries to trigger field collisions, but it does not raise the cost of natural-match probing. Token and label rounds are independent (Proposition~\ref{prop:multi-round}) but are sent sequentially, so pipelining them for throughput is also future work.

\section{Conclusion}\label{sec:conclusion}

Standard per-trial analysis for \mbox{FLPSI} fails under realistic workloads, where a spurious accept is not a mere accuracy artifact but a soundness defect: the kernel accepts a query the plaintext matcher would reject, and returns a value it never would. We expose this workload dependence through a composable bound (Lemma~\ref{lem:far-bound}, Theorem~\ref{thm:multi-token}) and close the resulting gap with \mbox{CSTPSI}, whose cached query and garbled circuit keep its added security affordable: it holds the realization soundness error (RSE) within engineering thresholds at a million records with two token rounds, and the bound extends that guarantee to a billion records with a third. Our evaluation, on synthetic and real data, records no spurious accept under amplification.

\section{LLM Usage Statement}\label{sec:llm-usage}
The authors used an LLM solely for language editing, grammar correction, figure drafting assistance, table formatting, and minor presentation improvements. The LLM did not contribute any conceptual ideas, protocol designs, methodological innovations, experimental strategies, security analyses, or scientific insights. All intellectual and technical contributions originated entirely from the authors. Any material provided to the LLM for editing or formatting contained no sensitive, personal, confidential, or ethically restricted information.

\bibliographystyle{IEEEtran}
\bibliography{references}

@inproceedings{uzun2021fuzzy,
  title     = {Fuzzy Labeled Private Set Intersection with Applications to Private Real-Time Biometric Search},
  author    = {Uzun, Erkam and Chung, Simon P. and Kolesnikov, Vladimir and Boldyreva, Alexandra and Lee, Wenke},
  booktitle = {30th USENIX Security Symposium (USENIX Security 21)},
  pages     = {911--928},
  year      = {2021},
  publisher = {USENIX Association}
}

@article{shamir1979share,
  author    = {Shamir, Adi},
  title     = {How to Share a Secret},
  journal   = {Communications of the ACM},
  volume    = {22},
  number    = {11},
  pages     = {612--613},
  year      = {1979}
}

@inproceedings{chen2018labeled,
  author    = {Chen, Hao and Huang, Zhicong and Laine, Kim and Rindal, Peter},
  title     = {Labeled {PSI} from Fully Homomorphic Encryption with Malicious Security},
  booktitle = {Proceedings of the 2018 ACM SIGSAC Conference on Computer and Communications Security ({CCS})},
  pages     = {1223--1237},
  year      = {2018}
}

@inproceedings{Garimella:CRYPTO:2021,
  author    = {Garimella, Gayathri and Pinkas, Benny and Rosulek, Mike and Trieu, Ni and Yanai, Avishay},
  title     = {Oblivious Key-Value Stores and Amplification for Private Set Intersection},
  booktitle = {Advances in Cryptology -- {CRYPTO} 2021},
  series    = {Lecture Notes in Computer Science},
  volume    = {12828},
  pages     = {395--425},
  publisher = {Springer},
  year      = {2021},
  note      = {IACR ePrint 2021/883}
}

@inproceedings{cong2021labeled,
  author    = {Cong, Kelong and Moreno, Radames Cruz and da Gama, Mariana Botelho and Dai, Wei and Iliashenko, Ilia and Laine, Kim and Rosenberg, Michael},
  title     = {Labeled {PSI} from Homomorphic Encryption with Reduced Computation and Communication},
  booktitle = {Proceedings of the 2021 ACM SIGSAC Conference on Computer and Communications Security ({CCS})},
  pages     = {1135--1150},
  year      = {2021}
}

@inproceedings{ghosh2019threshold,
  author    = {Ghosh, Satrajit and Simkin, Mark},
  title     = {The Communication Complexity of Threshold Private Set Intersection},
  booktitle = {Advances in Cryptology - {CRYPTO} 2019},
  pages     = {3--29},
  year      = {2019}
}

@inproceedings{badrinarayanan2021multiparty,
  author    = {Badrinarayanan, Saikrishna and Miao, Peihan and Raghuraman, Srinivasan and Rindal, Peter},
  title     = {Multi-Party Threshold Private Set Intersection with Sublinear Communication},
  booktitle = {Public-Key Cryptography - {PKC} 2021},
  year      = {2021}
}

@misc{seal,
  title        = {{Microsoft SEAL (release 4.x)}},
  howpublished = {\url{https://github.com/microsoft/SEAL}},
  author       = {{Microsoft Research}},
  year         = {2023}
}

@inproceedings{aby,
  author    = {Demmler, Daniel and Schneider, Thomas and Zohner, Michael},
  title     = {{ABY}: A Framework for Efficient Mixed-Protocol Secure Two-Party Computation},
  booktitle = {22nd Annual Network and Distributed System Security Symposium ({NDSS})},
  year      = {2015}
}

@inproceedings{chandran2022circuit,
  author    = {Chandran, Nishanth and Gupta, Divya and Shah, Akash},
  title     = {Circuit-{PSI} with Linear Complexity via Relaxed Batch {OPPRF}},
  booktitle = {Proceedings on Privacy Enhancing Technologies ({PoPETs})},
  volume    = {2022},
  number    = {1},
  pages     = {353--372},
  year      = {2022}
}

@misc{wang2017emp,
  author       = {Wang, Xiao and Malozemoff, Alex J. and Katz, Jonathan},
  title        = {{EMP-toolkit}: Efficient MultiParty Computation Toolkit},
  howpublished = {\url{https://github.com/emp-toolkit}},
  year         = {2016}
}

@inproceedings{Pinkas:EUROCRYPT:2019,
  author    = {Pinkas, Benny and Schneider, Thomas and Tkachenko, Oleksandr and Yanai, Avishay},
  title     = {Efficient Circuit-Based {PSI} with Linear Communication},
  booktitle = {Advances in Cryptology -- {EUROCRYPT} 2019},
  series    = {Lecture Notes in Computer Science},
  volume    = {11478},
  pages     = {122--153},
  year      = {2019}
}

@misc{flint,
  title        = {{FLINT: Fast Library for Number Theory (release 3.4)}},
  howpublished = {\url{https://flintlib.org}},
  author       = {Hart, William and others},
  year         = {2024}
}

@misc{zeromq,
  title        = {{\O MQ (ZeroMQ): An open-source universal messaging library}},
  howpublished = {\url{https://zeromq.org}},
  year        = {2024}
}

@inproceedings{krssw2019,
  author    = {Kales, Daniel and Rechberger, Christian and Schneider, Thomas and Senker, Matthias and Weinert, Christian},
  title     = {Mobile Private Contact Discovery at Scale},
  booktitle = {28th {USENIX} Security Symposium ({USENIX} Security 19)},
  year      = {2019},
  publisher = {USENIX Association}
}

@inproceedings{rdfa2018,
  author    = {Resende, Amanda C. Davi and de Freitas Aranha, Diego},
  title     = {Faster Unbalanced Private Set Intersection},
  booktitle = {Financial Cryptography and Data Security ({FC} 2018)},
  series    = {Lecture Notes in Computer Science},
  year      = {2018},
  note      = {ePrint 2017/677},
  publisher = {Springer}
}

@inproceedings{cfr2023,
  author    = {Chakraborti, Anrin and Fanti, Giulia and Reiter, Michael K.},
  title     = {Distance-Aware Private Set Intersection},
  booktitle = {32nd {USENIX} Security Symposium ({USENIX} Security 23)},
  year      = {2023},
  publisher = {USENIX Association},
  note      = {arXiv:2112.14737}
}

@inproceedings{fnp2004,
  author    = {Michael J. Freedman and Kobbi Nissim and Benny Pinkas},
  title     = {Efficient Private Matching and Set Intersection},
  booktitle = {Advances in Cryptology -- {EUROCRYPT} 2004},
  series    = {Lecture Notes in Computer Science},
  volume    = {3027},
  pages     = {1--19},
  publisher = {Springer},
  year      = {2004}
}

@inproceedings{psz2014,
  author    = {Benny Pinkas and Thomas Schneider and Michael Zohner},
  title     = {Faster Private Set Intersection Based on {OT} Extension},
  booktitle = {23rd {USENIX} Security Symposium ({USENIX} Security 14)},
  pages     = {797--812},
  publisher = {{USENIX} Association},
  year      = {2014}
}

@inproceedings{kkrt2016,
  author    = {Vladimir Kolesnikov and Ranjit Kumaresan and Mike Rosulek and Ni Trieu},
  title     = {Efficient Batched Oblivious {PRF} with Applications to Private Set Intersection},
  booktitle = {Proceedings of the 2016 {ACM} {SIGSAC} Conference on Computer and Communications Security},
  pages     = {818--829},
  publisher = {{ACM}},
  year      = {2016}
}

@inproceedings{prty2020,
  author    = {Benny Pinkas and Mike Rosulek and Ni Trieu and Avishay Yanai},
  title     = {{PSI} from {PaXoS}: Fast, Malicious Private Set Intersection},
  booktitle = {Advances in Cryptology -- {EUROCRYPT} 2020},
  series    = {Lecture Notes in Computer Science},
  volume    = {12106},
  pages     = {739--767},
  publisher = {Springer},
  year      = {2020}
}

@inproceedings{rs2021,
  author    = {Peter Rindal and Phillipp Schoppmann},
  title     = {{VOLE-PSI}: Fast {OPRF} and Circuit-{PSI} from Vector-{OLE}},
  booktitle = {Advances in Cryptology -- {EUROCRYPT} 2021},
  series    = {Lecture Notes in Computer Science},
  volume    = {12697},
  pages     = {901--930},
  publisher = {Springer},
  year      = {2021}
}

@inproceedings{rr2022,
  author    = {Srinivasan Raghuraman and Peter Rindal},
  title     = {Blazing Fast {PSI} from Improved {OKVS} and Subfield {VOLE}},
  booktitle = {Proceedings of the 2022 {ACM} {SIGSAC} Conference on Computer and Communications Security},
  pages     = {2505--2517},
  publisher = {{ACM}},
  year      = {2022},
  note      = {ePrint 2022/320, \url{https://eprint.iacr.org/2022/320}}
}

@inproceedings{GRS22,
  author    = {Gayathri Garimella and Mike Rosulek and Jaspal Singh},
  title     = {Structure-Aware Private Set Intersection, with Applications to Fuzzy Matching},
  booktitle = {Advances in Cryptology -- {CRYPTO} 2022, Part {I}},
  series    = {Lecture Notes in Computer Science},
  volume    = {13507},
  pages     = {323--352},
  publisher = {Springer, Heidelberg},
  year      = {2022},
  note      = {ePrint 2022/1011, \url{https://eprint.iacr.org/2022/1011}}
}

@inproceedings{baarsenpu2024fuzzyhyperballs,
  author       = {Aron van Baarsen and Sihang Pu},
  title        = {Fuzzy Private Set Intersection with Large Hyperballs},
  booktitle    = {Advances in Cryptology -- {EUROCRYPT} 2024},
  editor       = {Marc Joye and Gregor Leander},
  series       = {Lecture Notes in Computer Science},
  volume       = {14655},
  pages        = {340--369},
  publisher    = {Springer},
  year         = {2024},
  doi          = {10.1007/978-3-031-58740-5_12},
  note         = {ePrint 2024/330, \url{https://eprint.iacr.org/2024/330}}
}

@inproceedings{gqllw2024fmap,
  author    = {Ying Gao and Lin Qi and Xiang Liu and Yuanchao Luo and Longxin Wang},
  title     = {Efficient Fuzzy Private Set Intersection from Fuzzy Mapping},
  booktitle = {Advances in Cryptology -- {ASIACRYPT} 2024},
  year      = {2024},
  note      = {Code: \url{https://github.com/ql70ql70/Fuzzy-Private-Set-Intersection-from-Fuzzy-Mapping}}
}

@inproceedings{pstkz2025darot,
  author    = {Lucas Piske and Jaspal Singh and Ni Trieu and Vladimir Kolesnikov and Vassilis Zikas},
  title     = {Distance-Aware {OT} with Application to Fuzzy {PSI}},
  booktitle = {Proceedings of the 2025 {ACM} {SIGSAC} Conference on Computer and Communications Security},
  series    = {{CCS} '25},
  pages     = {4679--4691},
  year      = {2025},
  publisher = {{ACM}},
  doi       = {10.1145/3719027.3744857},
  note      = {Full version: ePrint 2025/996}
}

@unpublished{zccbllww2025,
  author    = {Cong Zhang and Yu Chen and Yang Cao and Yujie Bai and Shuaishuai Li and Juntong Lin and Anyu Wang and Xiaoyun Wang},
  title     = {Fast Fuzzy {PSI} from Symmetric-Key Techniques},
  year      = {2025},
  note      = {IACR ePrint 2025/885, \url{https://eprint.iacr.org/2025/885}; venue not stated in manuscript}
}

@inproceedings{mwzl2026,
  author    = {Shengzhe Meng and Xiaodong Wang and Xv Zhou and Bei Liang},
  title     = {Unbalanced Fuzzy Private Set Intersection for {L}$_\infty$ Distance: Achieving Sublinear Communication with Large Set Size},
  booktitle = {35th USENIX Security Symposium (USENIX Security 26)},
  publisher = {USENIX Association},
  year      = {2026},
  note      = {Code: \url{https://doi.org/10.5281/zenodo.18223806}}
}

@inproceedings{yhwdwz2026fpsi,
  author    = {Xinpeng Yang and Meng Hao and Chenkai Weng and Robert H. Deng and Yonggang Wen and Tianwei Zhang},
  title     = {Efficient Fuzzy Private Set Intersection from Secret-shared {OPRF}},
  booktitle = {Proceedings of the 2026 {IEEE} Symposium on Security and Privacy ({S\&P})},
  year      = {2026},
  note      = {arXiv:2604.14909v1}
}

@inproceedings{boldyreva2014fuzzy,
  author    = {Boldyreva, Alexandra and Chenette, Nathan},
  title     = {Efficient Fuzzy Search on Encrypted Data},
  booktitle = {Fast Software Encryption ({FSE})},
  series    = {LNCS},
  volume    = {8540},
  pages     = {613--633},
  publisher = {Springer},
  year      = {2014}
}

@inproceedings{chen2020sanns,
  title     = {{SANNS}: Scaling Up Secure Approximate k-Nearest Neighbors Search},
  author    = {Chen, Hao and Chillotti, Ilaria and Dong, Yihe and Poburinnaya, Oxana and Razenshteyn, Ilya and Riazi, M. Sadegh},
  booktitle = {29th USENIX Security Symposium (USENIX Security 20)},
  pages     = {2111--2128},
  year      = {2020},
  publisher = {USENIX Association}
}

@article{bgv2014leveled,
  author    = {Brakerski, Zvika and Gentry, Craig and Vaikuntanathan, Vinod},
  title     = {{(Leveled)} Fully Homomorphic Encryption without Bootstrapping},
  journal   = {{ACM} Transactions on Computation Theory ({TOCT})},
  volume    = {6},
  number    = {3},
  pages     = {13:1--13:36},
  year      = {2014},
  note      = {Preliminary version: ITCS 2012}
}

@inproceedings{chen2017fast,
  author    = {Chen, Hao and Laine, Kim and Rindal, Peter},
  title     = {Fast Private Set Intersection from Homomorphic Encryption},
  booktitle = {Proceedings of the 2017 {ACM} {SIGSAC} Conference on Computer and Communications Security ({CCS})},
  pages     = {1243--1255},
  year      = {2017}
}

@inproceedings{fipr2005keyword,
  author    = {Freedman, Michael J. and Ishai, Yuval and Pinkas, Benny and Reingold, Omer},
  title     = {Keyword Search and Oblivious Pseudorandom Functions},
  booktitle = {Theory of Cryptography ({TCC})},
  series    = {LNCS},
  volume    = {3378},
  pages     = {303--324},
  publisher = {Springer},
  year      = {2005}
}

@inproceedings{psz2015phasing,
  author    = {Pinkas, Benny and Schneider, Thomas and Segev, Gil and Zohner, Michael},
  title     = {Phasing: Private Set Intersection Using Permutation-Based Hashing},
  booktitle = {24th {USENIX} Security Symposium ({USENIX} Security)},
  pages     = {515--530},
  year      = {2015}
}

@inproceedings{psww2018cuckoo,
  author    = {Pinkas, Benny and Schneider, Thomas and Weinert, Christian and Wieder, Udi},
  title     = {Efficient Circuit-Based {PSI} via Cuckoo Hashing},
  booktitle = {Advances in Cryptology -- {EUROCRYPT} 2018},
  series    = {LNCS},
  volume    = {10822},
  pages     = {125--157},
  publisher = {Springer},
  year      = {2018}
}

@article{smartvercauteren2014simd,
  author    = {Smart, Nigel P. and Vercauteren, Frederik},
  title     = {Fully Homomorphic {SIMD} Operations},
  journal   = {Designs, Codes and Cryptography},
  volume    = {71},
  number    = {1},
  pages     = {57--81},
  year      = {2014}
}

@inproceedings{ducasstehle2016sanitization,
  author    = {Ducas, L\'eo and Stehl\'e, Damien},
  title     = {Sanitization of {FHE} Ciphertexts},
  booktitle = {Advances in Cryptology -- {EUROCRYPT} 2016},
  series    = {LNCS},
  volume    = {9665},
  pages     = {294--310},
  publisher = {Springer},
  year      = {2016}
}

@misc{calapodescu2017compact,
  title={Compact fuzzy private matching using a fully-homomorphic encryption scheme},
  author={Calapodescu, Ioan and Estehghari, Saghar and Clier, Johan},
  year={2017},
  month=aug # "~29",
  publisher={Google Patents},
  note={US Patent 9,749,128}
}

@inproceedings{chmielewski2008fuzzy,
  title={Fuzzy private matching},
  author={Chmielewski, Lukasz and Hoepman, Jaap-Henk},
  booktitle={2008 Third International Conference on Availability, Reliability and Security},
  pages={327--334},
  year={2008},
  organization={IEEE}
}

@inproceedings{ye2009efficient,
  title={Efficient fuzzy matching and intersection on private datasets},
  author={Ye, Qingsong and Steinfeld, Ron and Pieprzyk, Josef and Wang, Huaxiong},
  booktitle={International Conference on Information Security and Cryptology},
  pages={211--228},
  year={2009},
  organization={Springer}
}

@inproceedings{barni2010privacy,
  title={Privacy-preserving fingercode authentication},
  author={Barni, Mauro and Bianchi, Tiziano and Catalano, Dario and Di Raimondo, Mario and Donida Labati, Ruggero and Failla, Pierluigi and Fiore, Dario and Lazzeretti, Riccardo and Piuri, Vincenzo and Scotti, Fabio and others},
  booktitle={Proceedings of the 12th ACM workshop on Multimedia and security},
  pages={231--240},
  year={2010}
}

@inproceedings{blanton2011secure,
  title={Secure and efficient protocols for iris and fingerprint identification},
  author={Blanton, Marina and Gasti, Paolo},
  booktitle={European Symposium on Research in Computer Security},
  pages={190--209},
  year={2011},
  organization={Springer}
}

@inproceedings{erkin2009privacy,
  title={Privacy-preserving face recognition},
  author={Erkin, Zekeriya and Franz, Martin and Guajardo, Jorge and Katzenbeisser, Stefan and Lagendijk, Inald and Toft, Tomas},
  booktitle={International symposium on privacy enhancing technologies symposium},
  pages={235--253},
  year={2009},
  organization={Springer}
}

@inproceedings{osadchy2010scifi,
  title={Scifi-a system for secure face identification},
  author={Osadchy, Margarita and Pinkas, Benny and Jarrous, Ayman and Moskovich, Boaz},
  booktitle={2010 IEEE Symposium on Security and Privacy},
  pages={239--254},
  year={2010},
  organization={IEEE}
}

@inproceedings{sadeghi2009efficient,
  title={Efficient privacy-preserving face recognition},
  author={Sadeghi, Ahmad-Reza and Schneider, Thomas and Wehrenberg, Immo},
  booktitle={International conference on information security and cryptology},
  pages={229--244},
  year={2009},
  organization={Springer}
}

@inproceedings{babenko2016efficient,
  title={Efficient Indexing of Billion-Scale Datasets of Deep Descriptors},
  author={Babenko, Artem and Lempitsky, Victor},
  booktitle={IEEE Conference on Computer Vision and Pattern Recognition (CVPR)},
  pages={2055--2063},
  year={2016}
}

@techreport{huang2008labeled,
  title={Labeled Faces in the Wild: A Database for Studying Face Recognition in Unconstrained Environments},
  author={Huang, Gary B. and Ramesh, Manu and Berg, Tamara and Learned-Miller, Erik},
  number={07-49},
  institution={University of Massachusetts, Amherst},
  year={2007}
}

@inproceedings{dang2025fpsi,
  author = {Dang, C. and Zhou, X. and Liang, B.},
  title = {Efficient Fuzzy {PSI} Based on Prefix Representation},
  booktitle = {Proceedings of the 2025 {ACM} {SIGSAC} Conference on Computer and Communications Security (CCS)},
  pages = {2204--2218},
  year = {2025},
  publisher = {{ACM}}
}

\appendices
\section{Query-Power Transmission Variants}\label{app:query-powers}

The sender's per-partition polynomial evaluation
$f_{t,i,p}(b_j) = \sum_{\ell=0}^{d-1} c_{t,i,p,\ell} \cdot
\mathrm{Enc}(b_j^\ell)$ with $d = s_{\mathit{part}}$ requires
ciphertexts $\mathrm{Enc}(b_j^\ell)$ for every $\ell \in [0, d)$.
Two variants populate this set; both are compatible with the
cross-round caching of \S\ref{sec:cstpsi-design}.

\subsection{Variant~A (full-receiver-side).}  The receiver knows
$b_j$ in plaintext (from the GC-OPRF of step~\ref{step:gc}),
computes $b_j^\ell$ for $\ell = 0, \dots, d-1$ modulo $F$,
encrypts each, and transmits the full set of $d$ ciphertexts.
Sender per-round work: $d$ \texttt{multiply\_plain} operations
at depth one.  No ciphertext-ciphertext multiplications.

\subsection{Variant~B (BSGS).}  The receiver transmits an
$O(\sqrt{d})$-element baby-step--giant-step sparse window
following APSI-style labeled \mbox{PSI}~\cite{cong2021labeled};
the sender expands to the full power set via $O(\sqrt{d})$
ciphertext-ciphertext multiplications plus relinearization at
session setup, consuming a small constant number of additional
\mbox{BFV} multiplicative levels.

\subsection{Trade-off.}  Variant~B saves $\sim\sqrt{d}\times$ in
communication (at $d = 32$, $\sim 0.24$~MiB vs $\sim 1.3$~MiB per
query-item bundle) but pays a session-setup tax: relinearization
in \mbox{SEAL} runs roughly $20{-}40\times$ a
\texttt{multiply\_plain} on Intel servers with AVX-512+AES-NI, and
disproportionately slower on hardware without those instructions
(ARM, AMD pre-Zen3 without AVX-512, embedded).  Variant~B's
extra depth also forces a longer \mbox{coeff\_modulus} chain,
paying a $\sim 1.3\times$ ciphertext-size overhead on every
subsequent operation.  Variant~B therefore dominates when the
network is the bottleneck and the hardware is Intel-AVX-tuned
(\mbox{WAN} deployments on modern servers); Variant~A dominates
under co-located benchmarks or non-Intel hardware.

\subsection{Reference artifact.}  The benchmarks of
\S\ref{sec:eval} co-locate sender and receiver over a
\mbox{ZeroMQ} loopback transport, where Variant~B's
communication savings are negligible while its cipher-cipher
session-setup tax remains.  The artifact therefore implements
Variant~A, aligning with the benchmarked regime.

\subsection{Caching contribution and follow-up.}  Independent of
the variant, \mbox{CSTPSI}'s cross-round caching of the power
bundle avoids $(T + K - 1)$ session-setup passes per query
session; under Variant~B this is $(T + K - 1) \cdot O(\sqrt{d})$
ciphertext-ciphertext multiplications avoided, on the order of
seconds at typical $K$ on the reference hardware.  A full
Variant~B implementation with a decoupled ablation flag and
cross-hardware measurement is queued for follow-up.

\section{Full Security Proof of \(\Pi_T\)}\label{app:security-proof}

This appendix gives the simulator constructions
\(\mathsf{Sim}_R, \mathsf{Sim}_S\) and the supporting propositions
that underpin the proof sketch of Theorem~\ref{thm:main-sec}.
Throughout, the real-world view \(\Real_{\Pi_T, P}\) and the
ideal-world view \(\Ideal_{\mathcal{F}_{\mathrm{CSTPSI}}, \mathsf{Sim}_P, P}\)
are as in Definition~\ref{def:cstpsi-security} of
\S\ref{sec:security}.  We proceed party by party: for each corrupted
party we describe the real view step by step (indicating which
steps of \(\Pi_T\) place messages into the view), construct the
simulator step by step (mirroring those messages), and argue
indistinguishability per step.  The substantive (soundness) term of
the simulator gap, the realization soundness error (RSE) composition
bound from Lemma~\ref{lem:far-bound} amplified by
Theorem~\ref{thm:multi-token}, falls out of the reconstruction step
on the receiver-side walk.  Two supporting propositions
(multi-round soundness and 1-GC transcript equivalence) are stated
at the end of the appendix and cited from inside the walks.

\subsection{Sender Privacy: Simulating the Receiver}\label{app:sec-sender}

\(\mathsf{Sim}_R\) takes as input \((\lambda, R, Y,
\mathcal{F}_{\mathrm{CSTPSI}}(Y, \mathit{Db}))\): the security
parameter, the corrupted party label, the receiver's query \(Y\),
and the ideal-world output (the multiset of matched labels).
\(R\)'s view consists of its input, output, randomness, and the
messages it receives; the simulator's task is to emulate those
messages consistently with the ideal output.

\paragraph{\textbf{Describing the receiver's view.}}  We walk through the steps of CSTPSI protocol \(\Pi_T\)
(Figure~\ref{fig:cstpsi-protocol}) to describe \(\View_R\).

Steps~\ref{step:init}--\ref{step:interp} are sender-internal: \(S\)
samples \(\kappa\), partitions, blinds with \(\mathrm{AES}_\kappa\),
constructs Shamir shares, and interpolates and SIMD-packs the
per-partition polynomials.  \(R\) receives no messages, so
\(\mathsf{Sim}_R\) does nothing to simulate them.

In Step~\ref{step:gc}, \(R\) and \(S\) run a single 1-GC AES
execution: \(R\) inputs \((y_1, \ldots, y_N)\), \(S\) inputs
\(\kappa\), and \(R\) obtains \(b_j = \mathrm{AES}_\kappa(y_j)
\bmod F\) (with the \(0 \mapsto 1\) convention) for every
\(j \in [N]\).  \(R\) receives GC messages whose transcript is
computationally indistinguishable from a simulated transcript by
the security of the underlying semi-honest GC
protocol~\cite{wang2017emp}.  The GC output
\(\{b_j\}_{j \in [N]}\) enters \(\View_R\) and must be emulated.
By Proposition~\ref{prop:1-gc}, the single 1-GC execution is
transcript-equivalent to the naive \((T+K)\)-GC composition, so
this analysis suffices for the entire session.

In Step~\ref{step:onlinequery}, \(R\) prepares the BSGS sparse
window \(\{\mathrm{Enc}(b_j^\ell) : \ell \in W\}\) under its own
BFV public key and transmits it once.  \(R\) sends here but does
not receive; \(S\)'s window expansion is server-internal and not
part of \(\View_R\).

In Step~\ref{step:hom}, for each of the \(T + K\) rounds \(R\)
receives BFV ciphertexts that decrypt to the per-partition
polynomial evaluations \(f_{t,i,p}(b_j)\) for every
\((i, j, p)\).  These plaintext reconstructions are the
informative messages \(\mathsf{Sim}_R\) must emulate, and they
are the locus of the RSE (soundness) term of the
simulator gap.

In Steps~\ref{step:reconstruct-token}--\ref{step:reconstruct-label},
\(R\) combines the decrypted evaluations into the per-pair
reconstructions \(\mathsf{KR}(i, j, t)\) (the two-point Lagrange
combination of Step~\ref{step:reconstruct-token}, instantiated at
\(k = 2\)) and outputs labels.
\(R\) receives no further messages; the work is local to its
view.

For sender privacy, \(R\) must learn nothing about non-matching
database rows beyond the ideal-world output.  The key observation
is that for any non-matching pair \((i, j)\), under the PRF
security of \(\mathrm{AES}_\kappa\) and the threshold-\(k = 2\)
property of Shamir secret sharing, the real
\(\mathsf{KR}(i, j, t)\) is uniform on \(\mathbb{F}_F\); the
sole exception is the RSE composition event where a non-match
collides to zero in all \(T\) token rounds, bounded by
Lemma~\ref{lem:far-bound}.  This exception is the soundness term
of the simulator gap, not a leak: on it the value the receiver
reconstructs is an off-curve field element independent of the
stored records (a two-point Lagrange value at \(x = 0\) with the
token channel constrained to \(0\) and the record channel free),
so even there no sender data is revealed and the defect is one of
soundness.

\paragraph{\textbf{Constructing the receiver's simulator.}}  We mirror the
real view step by step.

For Steps~\ref{step:init}--\ref{step:interp}, \(\mathsf{Sim}_R\)
does nothing (no messages to \(R\)).

For Step~\ref{step:gc}, \(\mathsf{Sim}_R\) invokes the
semi-honest GC simulator~\cite{wang2017emp} on receiver-side
input \(Y\) and a freshly sampled simulated output
\(\widetilde{b}_j \getsr \mathbb{F}_F\) for each \(j \in [N]\),
yielding a simulated GC transcript together with simulated
blinded items \(\{\widetilde{b}_j\}\) used downstream.

For Step~\ref{step:onlinequery}, \(\mathsf{Sim}_R\) generates the
BSGS power bundle exactly as \(R\) would on inputs
\(\{\widetilde{b}_j\}\) using fresh BFV encryption randomness.

For Step~\ref{step:hom}, \(\mathsf{Sim}_R\) populates the
decrypted evaluations \(\mathsf{KR}(i, j, t)\) using the
ideal-world output: for each unordered pair \(\{i, j\}\), each
partition \(p\), and each round \(t\),
\begin{itemize}
  \item if \(\mathcal{F}_{\mathrm{CSTPSI}}\) reports a match at
    \((p, \{i, j\})\), set \(\mathsf{KR}(i, j, t) = 0\) in
    every token round \(t \in [0, T)\) and to the corresponding
    Lagrange-encoded label chunk for every label round
    \(t \in [T, T + K)\), consistent with the reported label;
  \item otherwise sample \(\mathsf{KR}(i, j, t) \getsr
    \mathbb{F}_F\) independently across \(t\).
\end{itemize}
\(\mathsf{Sim}_R\) then constructs simulated BFV ciphertexts
encrypting plaintext polynomials whose batched coefficients
decrypt to these values under \(R\)'s public key, using fresh
encryption randomness; by Proposition~\ref{prop:multi-round} the
per-round ciphertexts can be simulated independently across the
\(T + K\) rounds.

For Steps~\ref{step:reconstruct-token}--\ref{step:reconstruct-label},
\(R\) does only local computation; \(\mathsf{Sim}_R\) runs the
same computation on the simulated values, which by construction
yields the ideal-world output.

\paragraph{\textbf{Indistinguishability}} We argue step by step to show the indistinguishability of \(\Real_{\Pi_T, R}\) and
\(\Ideal_{\mathcal{F}_{\mathrm{CSTPSI}}, \mathsf{Sim}_R, R}\).

\emph{Steps \ref{step:init}--\ref{step:interp}.}  \(R\) receives
nothing; the two views are identical.

\emph{Step \ref{step:gc}.}  The real GC transcript on inputs
\((Y, \kappa)\) is computationally indistinguishable from the
simulated transcript by the semi-honest GC
simulator~\cite{wang2017emp}.  The real outputs
\(b_j = \mathrm{AES}_\kappa(y_j) \bmod F\) are PRF-uniform on
\(\mathbb{F}_F\) under the AES PRF assumption, up to negligible
mod-\(F\) bias; the simulated \(\widetilde{b}_j\) are exactly
uniform.  The gap on this step is bounded by the GC simulator
advantage plus the AES PRF advantage of the single GC execution,
amortized across all rounds by Proposition~\ref{prop:1-gc}; the
lemma's PRF hybrid charges \(T \cdot
\mathsf{Adv}^{\mathrm{prf}}_{\mathrm{AES}}\) in total across the
\(T\) token rounds.

\emph{Step \ref{step:onlinequery}.}  \(R\) sends; the sent
ciphertexts in both worlds are generated from the (then-fixed)
blinded items under fresh BFV encryption randomness.
Conditioned on Step~\ref{step:gc} indistinguishability, the two
bundles are identically distributed up to BFV randomness.

\emph{Step \ref{step:hom}.}  The substantive gap.  We split by
pair type.

For matching pairs \((p, \{i, j\})\) consistent with the
ideal-world output, the real reconstructions
\(\mathsf{KR}(i, j, t)\) are determined by the Shamir scheme:
\(s_{t,e} = 0\) for token rounds and the label chunk for label
rounds.  These match \(\mathsf{Sim}_R\)'s planted values exactly.

For non-matching pairs \((p, \{i, j\})\), the real
\(\mathsf{KR}(i, j, t)\) at every round is uniform on
\(\mathbb{F}_F\) conditioned on \(R\)'s ideal-world output (by
the threshold-\(k = 2\) property of Shamir secret sharing
combined with PRF-uniformity of \(\{b_j\}\) from
Step~\ref{step:gc}); the simulated values are uniform by
construction.  We argue this uniformity marginally per
non-matching pair, which suffices for the semi-honest sketch in
the style of FLPSI'21~\cite{uzun2021fuzzy}; the within-partition
correlations among the \(\binom{N}{k}\) reconstructions that
share one low-degree interpolant are not needed for the bound and
are not exploited by the simulator.  The two distributions agree
on non-matches except
on the RSE composition event where the real reconstruction
collides to zero across all \(T\) token rounds at some
non-matching pair (mis-classified by \(R\) as a token candidate).
The probability of this event is at most
\(\binom{N}{k}\,n_{\mathit{part}} / F^T\): the per-trial
single-round collision probability \(1/F\) is driven down to
\(1/F^T\) over the \(T\) independent token rounds
(Theorem~\ref{thm:multi-token}), and this is then union-bounded
over the \(\binom{N}{k}\,n_{\mathit{part}}\) trials per query
(Lemma~\ref{lem:far-bound}).

The BFV ciphertexts carrying these values are computationally
indistinguishable from fresh encryptions by IND-CPA security of
BFV~\cite{seal}.  By Proposition~\ref{prop:multi-round}, the
ciphertext bundle is transmitted once and reused, so the
ciphertext-level gap is a single IND-CPA advantage rather than
\(T + K\) of them.

\emph{Steps \ref{step:reconstruct-token}--\ref{step:reconstruct-label}.}
Local computation on values that are now indistinguishable; the
outputs agree by construction.

Summing the per-step gaps, the receiver-side simulator gap is
\[
  \binom{N}{k}\,\frac{n_{\mathit{part}}}{F^{T}}
  \,+\, T \cdot \mathsf{Adv}^{\mathrm{prf}}_{\mathrm{AES}}(\mathcal{B})
  \,+\, \mathsf{negl}(\lambda),
\]
where the negligible term collects the GC simulator advantage and
the BFV IND-CPA advantage.  At the parameters of
Corollary~\ref{cor:sufficient-T}, the soundness term is below the
engineering RSE threshold for \(T \ge 2\) at million-scale
and below the cryptographic standard for \(T \ge 3\) up to
ten-million scale.  It reaches the cryptographic standard at
billion-scale only at \(T = 4\)
(Corollary~\ref{cor:sufficient-T}).

\subsection{Receiver Privacy: Simulating the Sender}\label{app:sec-receiver}

\(\mathsf{Sim}_S\) takes as input \((\lambda, S, \mathit{Db},
|Y|)\): the security parameter, the corrupted party label, the
sender's database, and the leakage \(|Y|\) (which equals \(N\)
and is public).  \(\mathcal{F}_{\mathrm{CSTPSI}}\) returns
\(\bot\) to \(S\), so no ideal-world output enters the simulator.

\paragraph{\textbf{Describing the sender's view.}}  We walk through the steps of CSTPSI protocol \(\Pi_T\)
(Figure~\ref{fig:cstpsi-protocol}) to describe \(\View_S\).

Steps~\ref{step:init}--\ref{step:interp} are sender-internal:
\(S\) computes on its own inputs and randomness; no messages
arrive.  \(\mathsf{Sim}_S\) follows the protocol honestly on its
own input \(\mathit{Db}\); no simulation work is required.

In Step~\ref{step:gc}, \(S\) participates in the 1-GC AES
execution with input \(\kappa\) and output \(\bot\).  \(S\)
receives GC messages whose transcript is computationally
indistinguishable from a simulated transcript by the semi-honest
GC simulator~\cite{wang2017emp}.  By Proposition~\ref{prop:1-gc},
the single GC transcript is identical in distribution to the
marginal of \(T + K\) independent GC executions on the same
inputs, so this analysis covers the full session.

In Step~\ref{step:onlinequery}, \(S\) receives the receiver's
encrypted BSGS sparse window
\(\{\mathrm{Enc}(b_j^\ell)\}_{j,\,\ell \in W}\): \(|W| \cdot N\)
BFV ciphertexts under \(R\)'s public key.  This is the only
ciphertext message \(S\) receives in the entire session.  \(S\)'s
subsequent window expansion and per-round homomorphic
evaluations are server-internal: \(S\) produces the response
ciphertexts but does not receive them, so they are not part of
\(\View_S\).  By Proposition~\ref{prop:multi-round}, the
receiver's ciphertext bundle is transmitted once and reused
across the \(T + K\) rounds; the sender-side ciphertext-level gap
is therefore a single IND-CPA advantage rather than \(T + K\) of
them.

In Steps~\ref{step:hom} and
\ref{step:reconstruct-token}--\ref{step:reconstruct-label}, \(S\)
either acts internally (Step~\ref{step:hom}) or is not involved
(reconstruction is receiver-local).  \(\View_S\) gains no new
messages.

\paragraph{\textbf{Constructing the sender's simulator.}}  We again mirror
the real view step by step.

For Steps~\ref{step:init}--\ref{step:interp}, \(\mathsf{Sim}_S\)
runs the protocol honestly on the sender's own input
\(\mathit{Db}\); the resulting offline state is exactly what
\(S\) would compute in the real world.

For Step~\ref{step:gc}, \(\mathsf{Sim}_S\) invokes the
semi-honest GC simulator~\cite{wang2017emp} on sender-side input
\(\kappa\) and output \(\bot\), producing a simulated GC
transcript.

For Step~\ref{step:onlinequery}, \(\mathsf{Sim}_S\) samples
\(|W| \cdot N\) random plaintexts in \(\mathbb{F}_F\) of the
appropriate batched length and encrypts them under fresh BFV
randomness with \(R\)'s public key, producing a simulated query
bundle.

For Steps~\ref{step:hom} and
\ref{step:reconstruct-token}--\ref{step:reconstruct-label},
\(\mathsf{Sim}_S\) does nothing additional: no messages are
received by \(S\) in these steps.

\paragraph{\textbf{Indistinguishability}} We argue step by step to show the indistinguishability of \(\Real_{\Pi_T, S}\) and
\(\Ideal_{\mathcal{F}_{\mathrm{CSTPSI}}, \mathsf{Sim}_S, S}\).

\emph{Steps \ref{step:init}--\ref{step:interp}.}  Honest execution
on the sender's own input matches the real view exactly.

\emph{Step \ref{step:gc}.}  Computational indistinguishability of
the GC transcript follows from the semi-honest GC simulator
guarantee~\cite{wang2017emp}.  By Proposition~\ref{prop:1-gc},
this covers the full \((T + K)\)-round session at the cost of
one GC simulator invocation.

\emph{Step \ref{step:onlinequery}.}  Any distinguisher between
the real query-power bundle and the simulated random-plaintext
bundle yields a BFV IND-CPA distinguisher of the same
advantage~\cite{seal}, since both bundles consist of BFV
ciphertexts under \(R\)'s public key and differ only in their
plaintexts.  By Proposition~\ref{prop:multi-round}, the bundle is
transmitted once, so a single IND-CPA reduction suffices for the
entire session.

\emph{Steps \ref{step:hom} and
\ref{step:reconstruct-token}--\ref{step:reconstruct-label}.}  No
messages enter \(\View_S\); the two views agree trivially.

The sender-side simulator gap reduces to the BFV IND-CPA
advantage, the semi-honest GC simulator advantage, and the
information-theoretic threshold property of Shamir secret
sharing~\cite{shamir1979share}, summing to \(\mathsf{negl}(\lambda)\).

\subsection{Supporting Propositions}\label{app:sec-rounds}

The two propositions below are invoked from inside both
step-walks above.

\begin{proposition}[Multi-round security]\label{prop:multi-round}
The transcript of the \(T + K\) homomorphic-evaluation rounds in
\(\Pi_T\) leaks no more, up to a negligible advantage, than what
a sender or receiver would learn from \(T + K\) independent
semi-honest executions of the single-token-round configuration
\(\Pi_1\).
\end{proposition}

\begin{proof}[Argument]
The receiver's encrypted query powers are transmitted exactly
once at the start of the online phase
(Step~\ref{step:onlinequery}) and reused as ciphertext operands
in every subsequent round (Step~\ref{step:hom}).  Their inclusion
in the sender's view contributes only to the first round;
subsequent rounds' use of the same ciphertext bundle adds nothing
under the BFV IND-CPA assumption applied to the single
transmission.  The Shamir polynomials \(\{P_{t,e}\}_t\) for
distinct \(t\) (Step~\ref{step:share}) are sampled with
independent fresh coefficients
(Theorem~\ref{thm:multi-token}, conditional on a cryptographically secure pseudorandom number generator (CSPRNG) as the randomness source);
the joint distribution of \(\{f_{t,i,p}\}_t\) therefore factors
as the product of per-round distributions, and the receiver's
view across rounds factors correspondingly into independent
single-round views.  Sender-side, no per-round ciphertext is
added to \(\View_S\), so no additional simulation gap arises.
\qed
\end{proof}

\begin{proposition}[1-GC transcript identity]\label{prop:1-gc}
The transcript of the single GC execution in
Step~\ref{step:gc} of \(\Pi_T\) is identical in distribution to
the transcript of \(T + K\) independent GC executions on the
same input \((y_1, \ldots, y_N)\) and key \(\kappa\).  Hence the
optimization is transcript-equivalent to the naive
\((T + K)\)-GC composition and inherits its semi-honest security
without additional reductions.
\end{proposition}

\begin{proof}[Argument]
The blinded query items \(b_j = \mathrm{AES}_\kappa(y_j) \bmod F\)
with the \(0 \mapsto 1\) remap are deterministic in
\((y_1, \ldots, y_N, \kappa)\); evaluating the same circuit on
the same inputs yields identical outputs whether executed once
or \(T + K\) times.  The single-execution transcript is therefore
a marginal of the multi-execution distribution and is simulated
by the standard semi-honest GC simulator.
\qed
\end{proof}

\section{Protocol Walkthrough}\label{app:protocol-diagram}

Figure~\ref{fig:cstpsi-sequence} gives a message-sequence view of
\mbox{CSTPSI} ($\Pi_T$, Fig.~\ref{fig:cstpsi-protocol}) across the
database, sender, and receiver.

\begin{figure*}[p]
  \centering
  \includegraphics[width=\textwidth,height=0.92\textheight,keepaspectratio]%
    {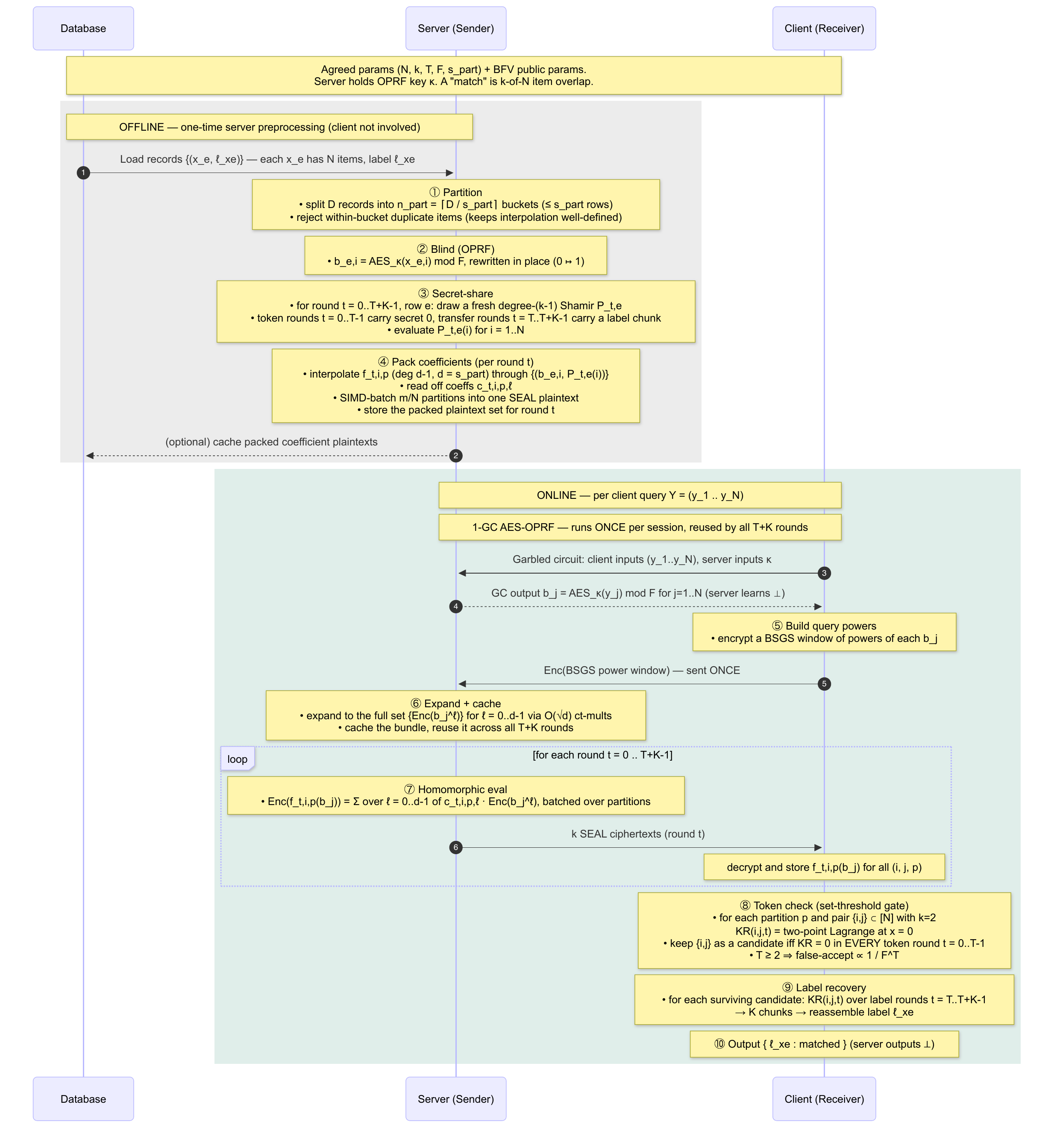}
  \caption{Message-sequence walkthrough of \mbox{CSTPSI} ($\Pi_T$,
    Fig.~\ref{fig:cstpsi-protocol}) with all optimizations enabled and
    $T\ge2$. \textbf{Offline} (grey) is sender-only and one-time:
    partition, OPRF-blind, secret-share over $T+K$ rounds, and pack the
    interpolated coefficients into SIMD-batched \mbox{SEAL} plaintexts.
    \textbf{Online} (teal) amortizes a single garbled circuit and a
    single cached query-power transmission across all $T+K$ rounds; only
    the per-round homomorphic evaluation repeats. A pair is admitted
    only if it reconstructs the token in \emph{every} one of the $T$
    token rounds, yielding the $\propto 1/F^{T}$ realization soundness
    error (RSE) bound of Lemma~\ref{lem:far-bound}.}
  \label{fig:cstpsi-sequence}
\end{figure*}

\end{document}